\documentclass[3p,authoryear,preprint,12pt]{elsarticle}
\long\def\twocolumn[#1]{#1}
\makeatletter\if@twocolumn\PassOptionsToPackage{switch}{lineno}\else\fi\makeatother

\usepackage{tabulary,xcolor}
\usepackage{amsfonts,amsmath,amssymb}
\usepackage[T1]{fontenc}
\makeatletter
\let\save@ps@pprintTitle\ps@pprintTitle
\def\ps@pprintTitle{\save@ps@pprintTitle\gdef\@oddfoot{\footnotesize\itshape \null\hfill\today}}
\def\hlinewd#1{%
	\noalign{\ifnum0=`}\fi\hrule \@height #1%
	\futurelet\reserved@a\@xhline}

\AtBeginDocument{\ifNAT@numbers \biboptions{sort&compress}\fi}
\makeatother

\usepackage{ifluatex}
\ifluatex
\usepackage{fontspec}
\defaultfontfeatures{Ligatures=TeX}
\usepackage[]{unicode-math}
\unimathsetup{math-style=TeX}
\else 
\usepackage[utf8]{inputenc}
\fi 
\ifluatex\else\usepackage{stmaryrd}\fi

\usepackage{url,multirow,morefloats,floatflt,cancel,tfrupee}
\makeatletter

\AtBeginDocument{\@ifpackageloaded{textcomp}{}{\usepackage{textcomp}}}
\makeatother
\usepackage{colortbl}
\usepackage{xcolor}
\usepackage{pifont}
\usepackage[nointegrals]{wasysym}
\makeatletter

\def\mcWidth#1{\csname TY@F#1\endcsname+\tabcolsep}

\def\cAlignHack{\rightskip\@flushglue\leftskip\@flushglue\parindent\z@\parfillskip\z@skip}
\def\rAlignHack{\rightskip\z@skip\leftskip\@flushglue \parindent\z@\parfillskip\z@skip}

\if@twocolumn\usepackage{dblfloatfix}\fi 
\AtBeginDocument{
	\expandafter\ifx\csname eqalign\endcsname\relax
	\def\eqalign#1{\null\vcenter{\def\\{\cr}\openup\jot\m@th
			\ialign{\strut$\displaystyle{##}$\hfil&$\displaystyle{{}##}$\hfil
				\crcr#1\crcr}}\,}
	\fi
}

\AtBeginDocument{%
	\@ifpackageloaded{endfloat}%
	{\renewcommand\efloat@iwrite[1]{\immediate\expandafter\protected@write\csname efloat@post#1\endcsname{}}}{}%
}%

\let\lt=<
\let\gt=>
\def\processVert{\ifmmode|\else\textbar\fi}

\@ifundefined{subparagraph}{
	\def\subparagraph{\@startsection{paragraph}{5}{2\parindent}{0ex plus 0.1ex minus 0.1ex}%
		{0ex}{\normalfont\small\itshape}}%
}{}

\newcommand\role[1]{\unskip}
\newcommand\aucollab[1]{\unskip}

\@ifundefined{tsGraphicsScaleX}{\gdef\tsGraphicsScaleX{1}}{}
\@ifundefined{tsGraphicsScaleY}{\gdef\tsGraphicsScaleY{.9}}{}
\def\checkGraphicsWidth{\ifdim\Gin@nat@width>\linewidth
	\tsGraphicsScaleX\linewidth\else\Gin@nat@width\fi}

\def\checkGraphicsHeight{\ifdim\Gin@nat@height>.9\textheight
	\tsGraphicsScaleY\textheight\else\Gin@nat@height\fi}

\def\fixFloatSize#1{}
\let\ts@includegraphics\includegraphics

\def\inlinegraphic[#1]#2{{\edef\@tempa{#1}\edef\baseline@shift{\ifx\@tempa\@empty0\else#1\fi}\edef\tempZ{\the\numexpr(\numexpr(\baseline@shift*\f@size/100))}\protect\raisebox{\tempZ pt}{\ts@includegraphics{#2}}}}

\AtBeginDocument{\def\includegraphics{\@ifnextchar[{\ts@includegraphics}{\ts@includegraphics[width=\checkGraphicsWidth,height=\checkGraphicsHeight,keepaspectratio]}}}

\def\URL#1#2{\@ifundefined{href}{#2}{\href{#1}{#2}}}

\def\UrlOrds{\do\*\do\-\do\~\do\'\do\"\do\-}%
\g@addto@macro{\UrlBreaks}{\UrlOrds}

\@ifundefined{quoteAttrib}
{}
{}

\@ifundefined{titlequoteAttrib}
{}{}

\newenvironment{title-quote}
{\list{}{\fontsize{10pt}{12pt}\selectfont\leftmargin.5in\itshape\rightmargin\leftmargin}%
	\item\relax}
{\endlist}

\makeatother

\usepackage{lmodern}			
\usepackage[T1]{fontenc}		
\usepackage[utf8]{inputenc}		
\usepackage{lastpage}			
\usepackage{indentfirst}		
\usepackage{color,xcolor}		
\usepackage{graphicx}			
\usepackage{microtype} 			

\usepackage{amsthm}			
\usepackage{amsmath} 
\usepackage{cleveref}			

\usepackage{listings}           	
 \usepackage{pdfpages}           	

\usepackage{placeins} 
\usepackage{epsfig}
\usepackage{bm}
\usepackage{bbm}
\usepackage{amsfonts}
\usepackage{multirow}
\usepackage{subcaption} 
\usepackage{epstopdf}

\usepackage{breqn} 

\usepackage{amssymb}
\usepackage{multirow}

\theoremstyle{definition}
\newtheorem{thm}{Teorema}[section]

\newtheorem{dfn}[thm]{Definition}

\newtheorem{pps}[thm]{Proposition}

\newcommand{\das}{\stackrel{d}{=}}
\newcommand{\distras}[1]{%
	\savebox{\mybox}{\hbox{\kern3pt$\scriptstyle#1$\kern3pt}}%
	\savebox{\mysim}{\hbox{$\sim$}}%
	\mathbin{\overset{#1}{\kern\z@\resizebox{\wd\mybox}{\ht\mysim}{$\sim$}}}%
}

\newcommand{\nub}{\bm{\nu}}

\newcommand{\betab}{\bm{\beta}}
\newcommand{\Xb}{\bm{X}}

\journal{Journal of Statistical Planning and Inference}

\begin{document}
\begin{frontmatter}
	
\title{Regression models for binary data with scale mixtures of centered skew-normal link functions}

\author[]{João Victor B. de Freitas}
\author[]{Caio L. N. Azevedo}

\begin{abstract}
For the binary regression, the use of symmetrical link functions are not appropriate when we have evidence that the probability of success increases at a different rate than decreases. In these cases, the use of link functions based on the cumulative distribution function of a skewed and heavy tailed distribution can be useful. The most popular choice is some scale mixtures of skew-normal distribution. This family of distributions can have some identifiability problems, caused by the so-called direct parameterization. Also, in the binary modeling with skewed link functions, we can have another identifiability problem caused by the presence of the intercept and the skewness parameter. To circumvent these issues, in this work we proposed link functions based on the scale mixtures of skew-normal distributions under the centered parameterization. Furthermore, we proposed to fix the sign of the skewness parameter, which is a new perspective in the literature to deal with the identifiability problem in skewed link functions. Bayesian inference using MCMC algorithms and residual analysis are developed. Simulation studies are performed to evaluate the performance of the model. Also, the methodology is applied in a heart disease data. 
\end{abstract}
	
\end{frontmatter}

\section{Introduction}
Binary regression models are adequate to analyze data when the response variable assumes only two values. In these models the probability of success of a binary response is estimated based on one or more covariates through the specification of a link function. According to \cite{Chen99}, the degree of skewness of the link function can be measured by the rate at which the probability of success of a response variable approaches to 0 or 1. A link function is symmetric if the approximation rate of the probability of success for 0 is the same as the approximation rate for 1~\citep{Chen99}, such as the probit and logit link functions, in the same sense, a link function is positively skewed if the approximation rate of the probability of success for 1 is faster than the approximation rate for 0, and negatively skewed, otherwise. \cite{Czado} showed, through a simulation study, that the link misspecification, in terms of skewness, can lead to bias in the estimates of the regression coefficients. We can solve this by the use of skewed link functions, that can be obtained, for example, through the cumulative distribution function (CDF) of skewed distributions.

Many regression models with skewed link functions have been proposed in the literature. \cite{Stukel} proposed the generalized logistic models and \cite{Guerrero} proposed to use the Box-Cox transformation~\citep{box1964analysis} to the odds ratio to generalize the logistic regression model. Another option to deal with skewness in link functions is to consider the CDF of a skewed distribution. The most popular example of this method is the complementary log-log link function, which is constructed from the CDF of the Gumbel distribution. \cite{Chen99} proposed a skewed probit link, considering a class of mixture of normal distributions. \cite{Bazan2010} presented a unified approach for two skew probit links. In \cite{Bazan2014}, it was introduced two new skewed link functions, one based on the CDF of the power-normal distribution and another based on the CDF of the reciprocal power-normal distribution.

Since probit and logistic regression estimates are not robust in the presence of outliers, \cite{Liu2004} proposed a new binary model, named robit regression, in which the normal distribution in probit regression is replaced by a t-distribution with known or unknown degrees of freedom. Both the logistic model and the probit model can be approximated by the robit regression, as showed in \cite{Liu2004}. Instead using the t-distribution, \cite{Kim} introduced a class of skewed generalized t-link models, that accommodate heavy tails and skewness in link functions. Many of the proposals involving link functions based on CDF of skewed distributions use the approach of \cite{Albert}, however it is known that this approach can cause identifiability problems between the intercept and the skewness parameter~\citep{Chen99,Kim}.

Said that, the main contributions of this paper are:
\begin{enumerate}
    \item We developed a wide class of link functions for binary regression models that accommodates skewed and heavy tailed link functions, and includes the probit, skew probit, skew t, skew slash and skew contaminated normal models. This class is based on the scale mixtures of skew-normal (SMSN) family of distributions (see \cite{Clecio_VH_2011}) considering a centered parameterization of the skew-normal distribution, to avoid identifiability problems caused by the direct parameterization, namely scale mixtures of centered skew-normal (SMCSN) distributions.
    \item We propose to fix the sign of the skewness parameter to avoid identifiability problem with the intercept, and using simulations, we showed that this approach is efficient. This is a new perspective to deal with this identifiability problem, since the usual approaches are to exclude the intercept~\citep{Chen99} or to do a reparameterization~\citep{Kim}.
    \item We also discuss Bayesian inference and residual analysis for the proposed model.
    \item Simulation studies were performed to assess the behavior of
the MCMC algorithms to estimate the parameters.
    \item Analysis of a heart disease data showing that our
approach outperforms the probit and centered skew-normal link functions.
\end{enumerate}

\section{Scale mixtures of skew-normal distribution under the centered parametrization}

\subsection{The centered skew-normal distribution}
The scale mixtures of skew-normal distributions under the direct parametrization is constructed based on the skew-normal distribution, denoted by $Y \sim SN(\alpha,\beta ^2,\lambda)$. This distribution was originally introduced by \cite{Azzalini_1985}, whose probability density function (PDF) is given by $f(y|\alpha,\beta,\lambda)= 2\beta^{-1} \phi\left(\frac{y-\alpha}{\beta}\right)
\Phi\left(\lambda\left(\frac{y-\alpha}{\beta}\right)\right) I_{(-\infty,\infty)}(y)$, with location parameter  $\alpha \in \mathbb{R}$, scale parameter  $\beta \in \mathbb{R}^{+}$ and skewness parameter $\lambda \in \mathbb{R}$, where $\phi(.)$ and $\Phi(.)$ denote the standard normal PDF and CDF, respectively. 
As noticed by \cite{ArellanoAzzalini2008}, the direct parameterization of the skew-normal distribution has some identifiability problems, if $\lambda\approx 0$, the log-likelihood presents a non-quadratic shape. Even under the Bayesian paradigm, this fact can lead to some problems. 
\cite{Pewsey2000} addressed various issues related to direct parameterization and explained why it should not be used for estimation procedures. \cite{Azzalini_1985} noticed that when $\lambda\approx 0$ the Fisher Information is singular. More details of these discussions can be found in \cite{genton2004skew}. To circumvent this issue, \cite{Azzalini_1985} proposed an alternative parameterization, namely centered parameterization, which is defined by 
\begin{align}
Y = \mu + \sigma Z_0
\label{centred_sn}
\end{align}
where $Z_0 = \frac{Z - \mu_z}{\sigma_z}$ with $Z \sim SN(0,1,\lambda)$, $\mu_z = b\delta$ and $\sigma_z = \sqrt{1-b^2\delta^2}$. The centered parameterization is formed by the centered parameters $\mu \in \mathbb{R}$,  $ \sigma \in \mathbb{R}^{+} $ and $\delta \in (-1, 1)$, whose explicit expression are $\mu = E(Y) = \alpha  +\beta \mu_z$, $\sigma^2 = Var(Y) = \beta^2(1-\mu_z^2 )$ and $\delta = \lambda/\sqrt{1+\lambda^2}$. The centered parameterization of the skew-normal distribution, or centered skew-normal (CSN) distribution, will be denoted by $Y \sim CSN(\mu,\sigma^2,\delta)$. The density of  (\ref{centred_sn}), after some algebra, is given by
\begin{align}
f(y|\mu,\sigma^2,\delta) = 2\omega^{-1}\phi(\omega^{-1}(y-\xi))\Phi\left( \lambda \left( \frac{y-\xi}{\omega}\right) \right)
\label{density_centred_sn}
\end{align}
where $s = \left(\frac{2}{4-\pi}\right)^{1/3}$, $\xi = \mu - \sigma\gamma^{1/3}s$,
$\omega = \sigma \sqrt{1+s^2\gamma^{2/3}}$, $\lambda = \delta/\sqrt{1-\delta^2}$ and $\gamma = \frac{4 - \pi}{2} \frac{(b \delta)^3}{(1 - b^2 \delta^2)^{3/2}}$ denotes the Pearson's skewness coefficient.

Considering $Z\sim SN(0,1,\lambda)$, \cite{Henze1986} introduced a useful stochastic representation of this distribution, which is given by
\begin{align}
Y \das \delta H + \sqrt{1-\delta^2} T,
\label{stochastic_sn}
\end{align} 
where $\das$ means ``distributed as'', $H \sim HN(0,1) \bot T \sim N(0,1) $ and $HN(.)$ denotes the half-normal distribution.
Therefore, using (\ref{stochastic_sn}) and the CSN distribution as described in (\ref{centred_sn}), we have that the stochastic representation of the CSN, $Y \sim CSN(\mu,\sigma^2,\delta)$, is $Y \das \xi + \omega (\delta H + \sqrt{1-\delta^2} T)$, where $\xi$ and $\omega$ are defined as above.

\subsection{Scale mixture of centered skew-normal distributions}
Following the hierarchical representation of the scale mixture of skew-normal distribution under the direct parametrization described in \cite{Clecio_VH_2011}, we have the following definition.

\begin{dfn}
	A random variable Y follows a scale mixture of centered skew-normal distribution, or SMCSN distribution, if Y can be stochastically represented by
	\begin{align}
	Y \das \mu + k(U)^{1/2}Z
	\label{stochastic_smsn_cp},
	\end{align} 
	where $Z \sim CSN(0,\sigma^2,\delta)$ and U is a positive random variable with CDF  $G(.| \bm{\nu})$. 
	\label{def_smsn_cp}
\end{dfn} 
We use the notation $Y \sim SMCSN(\mu,\sigma^2,\delta,G,\nub)$ for a random variable represented as in Definition \ref{def_smsn_cp}.
From Definition \ref{def_smsn_cp}, it follows that $E(Y)=\mu$, since $E(Z)=0$, and $Var(Y)= \sigma^2 E(k(U))$. It also can be noticed that when $\delta=0$ we get the corresponding scale mixtures of normal distribution family, introduced by \cite{Andrews_Mallows}, since $Z \sim N(0,\sigma^2)$.
For this work, we will restrict this family considering $k(u)=\frac{1}{u}$. Under this restriction, we have the following examples of SMCSN distributions:
\begin{itemize}
	\item \textbf{Centered skew-t distribution:} this distribution is obtained considering $U \sim \mbox{gamma}(\nu/2,\nu/2)$, denoted by $CST(\mu,\sigma^2,\delta,\nu)$, where $\mu$ denotes the mean, $\delta$ the skewness parameter, $\nu$ the degree of freedom and $\sigma^2$ is related to the variance of Y through $Var(Y) = \sigma^2 \frac{\nu}{\nu -2} $, since $E(U^{-1}) = \frac{\nu}{\nu -2}$;
	\item \textbf{Centered skew-slash distribution:} this distribution is obtained considering $U \sim beta(\nu,1)$, denoted by $CSS(\mu,\sigma^2,\delta,\nu)$, where $\mu$ is the mean, $\delta$ the skewness parameter, $\nu$ the degree of freedom and $\sigma^2$ is related to the variance of Y through $Var(Y) = \sigma^2 \frac{\nu}{\nu -1}$, since $E(U^{-1}) = \frac{\nu}{\nu -1}$;
	\item \textbf{Centered skew contaminated normal distribution:} this distribution is obtained considering U a discrete random variable assuming only two values, with the following probability function $h(u|\bm{\nu}) = \nu_1 I(u=\nu_2) + (1-\nu_1)I(u=1)$ and $E(U^{-1}) =\frac{\nu_1 + \nu_2(1-\nu_1)}{\nu_2} $.
	This distribution is denoted by $CSCN(\mu,\sigma^2,\delta,\bm{\nu})$, where $\mu$ is the mean, $\delta$ is the skewness parameter and according to \cite{Aldo} the parameters $\nu_1$ and $\nu_2$ can be interpreted as the proportion of outliers and a scale factor,
	respectively. For this distribution, the variance of Y is equal to $\sigma^2 \frac{\nu_1 + \nu_2(1-\nu_1)}{\nu_2} $;
	\item \textbf{Centered skew-normal distribution:} this distribution is obtained considering $P(U=1)=1$;	
	\item \textbf{Normal distribution} this distribution is obtained considering $P(U=1)=1$ and $\delta=0$.
\end{itemize}

Using the stochastic representation in (\ref{stochastic_smsn_cp}), the CDF of this family can be written as
\begin{align}
F(y|\mu,\sigma^2,\delta,\nub) = \int_{-\infty}^y\left[\int_{0}^{\infty} f(x|\mu,\sigma^2 k(u),\delta) dG(u |\bm{\nu})\right]dx,
\label{cdf_smsn}
\end{align} 
where $f(x|\mu,\sigma^2 k(u),\delta)$ is the PDF of the CSN distribution as in (\ref{density_centred_sn}). For a SMCSN distribution, a hierarchical representation is useful to simplify the Bayesian estimation process. Based on the Definition \ref{def_smsn_cp}, this representation is given by
\begin{align}\label{smcsn_sto1}
\begin{split}
Y|U=u  &\sim CSN(\mu,\sigma^2 k(u), \delta),\\
U &\sim G(.| \bm{\nu}).
\end{split}
\end{align}

\section{Binary regression model} \label{s3: regbin}
Let $\bm{X} = (1,\bm{X}_1,\bm{X}_2,\dots,\bm{X}_{p-1})^\top$ be a $p\times n$ known design matrix of fixed covariates, $\bm{Y}=(Y_1,\dots,Y_n)^\top$ be a $n \times 1$ vector of dichotomous response variables, such that $y_i=1$ with probability $p_i$ and $y_i=0$ with probability $1-p_i$, and $\bm{\beta}= (\beta_0,\beta_1,\dots,\beta_{p-1})^\top$ be a $p \times 1$ vector of regression coefficients.
Consider the binary regression model assuming that $p_i = F(\eta_i) = F({X_i}^\top\bm{\beta}), i=1,\dots,n$, where $\eta_i ={X_i}^\top\bm{\beta} $, $F(\cdot)$ denotes a CDF and is a link function that represents the relationship between the probability of success and the covariates. In this paper, we assume that $p_i = F({X_i}^\top\bm{\beta}|\delta,\nub), i=1,\dots,n $,
where $F(\cdot|\delta,\nub)$ is the CDF of the SMCSN distribution with $\sigma^2=1$, $\delta$ is the skewness parameter and $\nub$ are shape parameters. 

The use of this distribution class in the binary regression model allows us a great flexibility in the choice of the link function, since this class includes heavy tailed, symmetric and skewed distributions. From Figure \ref{fig:link_stnssl}, we can see the effect of heavy tails on the CDFs by observing that the probability of success of the SMCSN CDFs grows slowly or fastly when compared to the CSN CDF. From these figures, we can also see that when $\delta = -0.95$ ($\delta = 0.95$) the probability $p_i$ approaches to 1 (0) at a faster rate than it approaches to 0 (1). When $\delta=0$ the probability of success approaches to 1 or 0 at the same rate. These figures suggest that the use of heavy tails distributions is appropriate in the cases where extreme values of the linear predictor are expected. In addition, heavy tails links help to control the rate of approximation of $p_i$ to 0 and 1, providing more flexibility in the modeling of the influence of the covariates in the response variable.

\begin{figure}[h!]
	\captionsetup[subfigure]{aboveskip=-3pt,belowskip=-3pt}
	\begin{subfigure}{\linewidth}
		\begin{subfigure}{\linewidth/2 - .5em}
			\includegraphics[width=\linewidth]{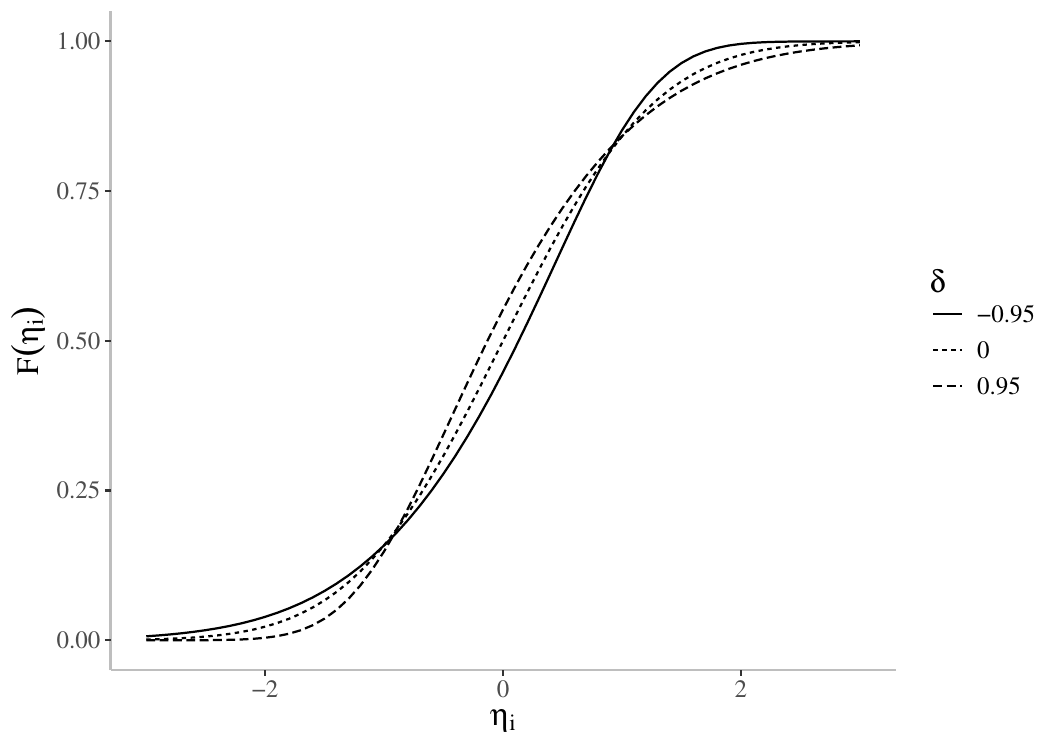}
			\caption{}
			\label{fig:link_csn}
		\end{subfigure}\hfill
		\begin{subfigure}{\linewidth/2 - .5em}
			\includegraphics[width=\linewidth]{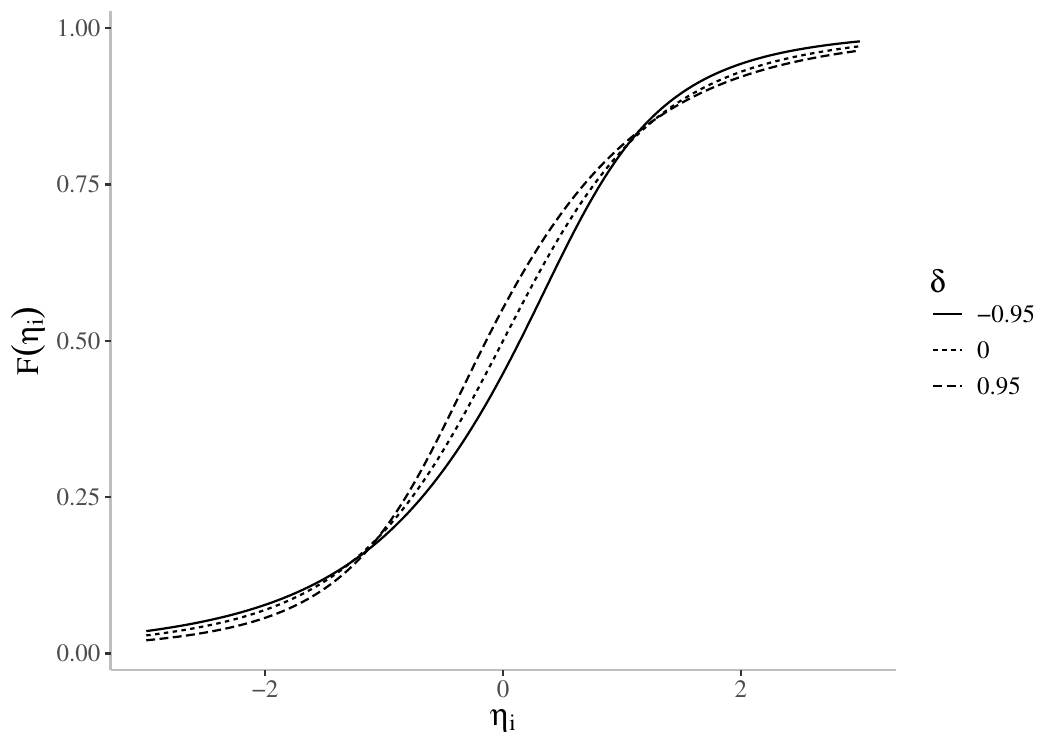}
			\caption{}
			\label{fig:link_cst}
		\end{subfigure}\hfill
		\begin{subfigure}{\linewidth/2 - .5em}
			\includegraphics[width=\linewidth]{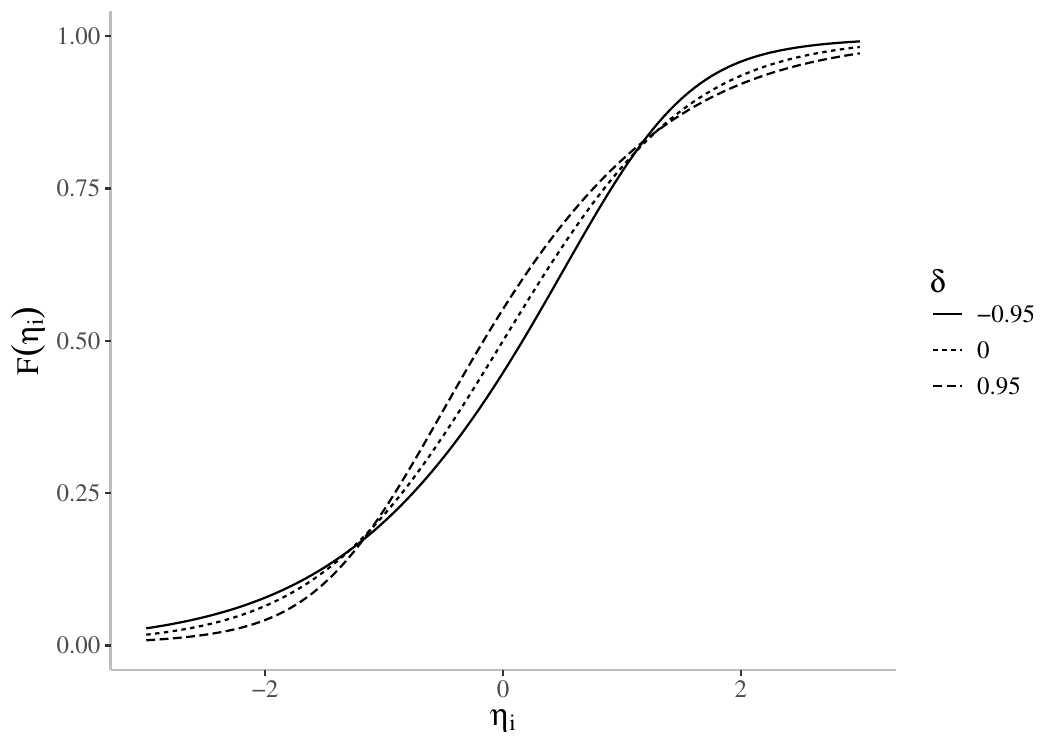}
			\caption{}
			\label{fig:link_css}
		\end{subfigure}
		\begin{subfigure}{\linewidth/2 - .5em}
			\includegraphics[width=\linewidth]{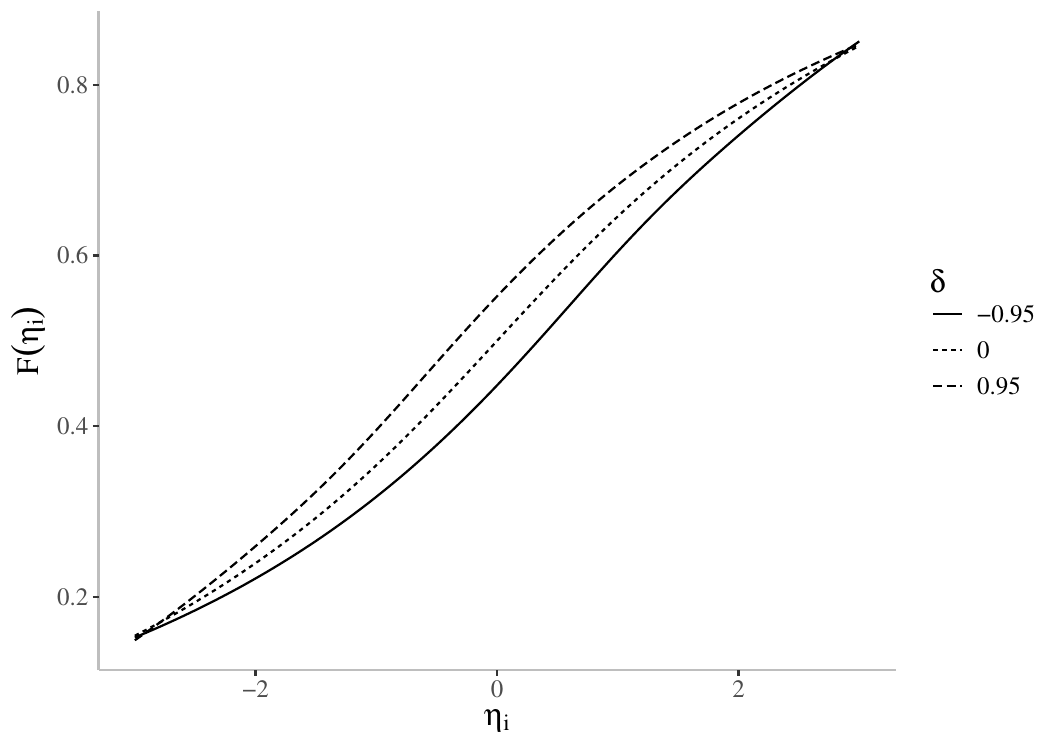}
			\caption{}
			\label{fig:link_cscn}
		\end{subfigure}
	\end{subfigure}
	\caption{Probability of success as a function of $\eta_i$ for various $\bm{\nu}$ for the CSN distribution (\subref{fig:link_csn}), CST distribution with $\nu=3$ (\subref{fig:link_cst}), CSS distribution with $\nu=2$(\subref{fig:link_css}) and CSCN distribution with $\bm{\nu} = (0.9,0.1)^\top$.}
	\label{fig:link_stnssl}
\end{figure}

To perform Bayesian inference, an approach based on data augmentation, as considered in \cite{Chib93} and \cite{Bazan2010}, will be used. The main advantage of this approach is the introduction of a hierarchical representation which simplifies the Bayesian estimation process. Said that, consider the following proposition.

\begin{pps}\label{prop3}
	The binary model $Y_i \sim Ber(p_i)$ with $p_i = F({X_i}^\top\bm{\beta}|\delta,\nub) $ is equivalent to consider
	\begin{align}\label{lat_bin_zf}
	Y_i = I(Z_i >0) = 
	\begin{cases}
	1 & \text{if $Z_i > 0$} \\
	0 & \text{if $Z_i\le 0$}
	\end{cases}, \quad i=1,\dots,n,
	\end{align} 
	with 
	\begin{align}\label{ziei_binf}
	\begin{split}
	Z_i &= {X_i}^\top\bm{\beta} - \Delta{U_i}^{-1/2}(H_i-b) + \sqrt{\tau}{U_i}^{-1/2}T_i,
	\end{split}   
	\end{align} 
 where $T_i \sim N(0,1)$, $H_i \sim HN(0,1)$, $U_i \sim G(.| \bm{\nu})$, $\Delta = \frac{\delta}{\sqrt{1-b^2\delta^2}}$ and $\tau = \frac{1-\delta^2}{1-b^2\delta^2}$.
 \end{pps} 

 \begin{proof}
     Considering $W_i \sim SMCSN(0,1,\delta,G,\nub)$, from \eqref{smcsn_sto1} we have that $W_i|U_i = u_i \sim CSN(0,1/u_i,\delta)$ with $U_i \sim G(.| \bm{\nu})$, and from the Henze stochastic representation it follows that $\xi_i = \delta^{1/3}s/\sqrt{u_i} = -b\Delta/\sqrt{u_i}$, $\omega_i = \sqrt{1+s^2\gamma^{2/3}}/\sqrt{u_i} = 1/(\sqrt{u_i}\sqrt{1-b^2\delta^2})$ and 
     \begin{align*}
         W_i|U_i = u_i \das -\frac{b\Delta}{\sqrt{u_i}} + \frac{1}{\sqrt{u_i}\sqrt{1-b^2\delta^2}}\left(\delta H_i + \sqrt{1-\delta^2}T_i\right) = \frac{1}{\sqrt{u_i}}[\Delta(H_i -b) + \sqrt{\tau} T_i],
     \end{align*}
     for $i=1,\dots,n$. Then,
	\begin{align}\label{wi-rep}
	\begin{split}
	W_i|H_i=h_i,U_i=u_i &\sim N(u_i^{-1/2}\Delta (h_i - b), u_i^{-1}\tau),\\
	H_i &\sim HN(0,1),\\
	U_i &\sim G(.| \bm{\nu}).
	\end{split}   
	\end{align}

	Therefore, considering \eqref{wi-rep}, $T_i \sim N(0,1)$, $\eta_i = {X_i}^\top\bm{\beta}$ and denoting the CDF of $H_i$ as $F(h_i)$, then
	\begin{align}
	\begin{split}
	p_i &= \int_0^{\infty} \int_0^{\infty} P\left(T_i \le \frac{\eta_i - u_i^{-1/2}\Delta (h_i - b)}{u_i^{-1/2}\sqrt{\tau}} |H_i,U_i\right) dG(u_i|\nub)dF(h_i)\\
	&= \int_0^{\infty} \int_0^{\infty} P\left(T_i > \frac{-\eta_i + u_i^{-1/2}\Delta (h_i - b)}{u_i^{-1/2}\sqrt{\tau}} |H_i,U_i\right) dG(u_i|\nub)dF(h_i)\\
	& = \int_0^{\infty} \int_0^{\infty} P\left(u_i^{-1/2}\sqrt{\tau} T_i + \eta_i - u_i^{-1/2}\Delta (h_i - b) >0 |H_i,U_i\right) dG(u_i|\nub)dF(h_i)\\  	
	& = \int_0^{\infty} \int_0^{\infty} P\left(Z_i >0 |H_i,U_i\right) dG(u_i|\nub)dF(h_i)\\  
	& =P(Z_i>0)	,
	\end{split}   
	\end{align} 
	where $Z_i \sim SMCSN(\eta_i,1,-\delta,G,\nub)$, for $i=1,\dots,n$. Therefore, it implies that considering $Y_i=I(Z_i>0)$ with $Z_i = {X_i}^\top\bm{\beta} + {U_i}^{-1/2}[\Delta(b-H_i) + \sqrt{\tau}T_i]$ is equivalent to consider $Y_i \sim \mbox{Ber}(p_i)$ with $p_i=F({X_i}^\top\bm{\beta}|\delta,\nub)$.
 \end{proof}

Following \eqref{lat_bin_zf} and \eqref{ziei_binf}, the hierarchical formulation of the model is given as follow:
\begin{align*}
\begin{split}
Z_i|U_i=u_i,H_i=h_i,y_i  &\sim N \left({X_i}^\top\bm{\beta} -{u_i}^{-1/2}\Delta(h_i-b),\frac{\tau}{u_i}\right)I(z_i,y_i), \\
H_i &\sim HN(0,1),\\
U_i &\sim G(.| \bm{\nu}),
\end{split}
\end{align*}
where $I(z_i,y_i)=I(z_i>0)I(y_1=1) + I(z_i\le 0)I(y_i=0)$.

Based on \cite{Kim}, we can see from \eqref{ziei_binf} that the intercept and $\Delta$ can be confounded with each other when we analyze the sign of $Z_i$. For example, in equation \eqref{lat_bin_zf} suppose that $Z_i = \beta_0 + \beta_1x_i - \Delta{U_i}^{-1/2}(H_i-b) + \sqrt{\tau}{U_i}^{-1/2}T_i = \beta_1x_i + [\beta_0 - \Delta{U_i}^{-1/2}(H_i-b)] + \sqrt{\tau}{U_i}^{-1/2}T_i = \beta_1x_i + \Delta^\ast + \sqrt{\tau}{U_i}^{-1/2}T_i$, we can see that $\Delta^\ast$ controls the skewness of the link function which depends on two parameters $\beta_0$ and $\delta$ and we are not able to know whether the sign of this skewness is controlled by $\beta_0$ or $\delta$, thus causing an identifiability problem. A way to handle with this issue is to consider a model without the intercept or do a reparameterization~\citep{Kim}. In this paper we propose another approach, if we fix the sign of the skewness parameter $\delta$ the identifiability problem is solved, as we will see in the simulation studies. This can be done by define a prior for $\delta$ in the interval $(0,1)$ or $(-1,0)$. In the next section, we show through simulations how to choose the sign of $\delta$ in practice.

\section{Bayesian Inference and simulation studies}\label{ss3:bayes}
To use the Bayesian paradigm, it is essential to obtain the joint posterior distribution. To obtain the posterior distribution, we need first to consider the complete likelihood
\begin{align}\label{completelikelihood}
L_c(\bm{\theta}|y,z,u,h) &\propto \prod_{i=1}^{n} \phi\left(z_i|\mu_i, \tau u_i^{-1}\right) I(z_i,y_i) f(h_i)h(u_i|\bm{\nu}) \nonumber\\
&\propto \prod_{i=1}^{n} \frac{\sqrt{u_i}}{\sqrt{\tau}} \exp\left\{ -\frac{u_i}{2\tau} \left(z_i - \mu_i \right)^2 \right\} I(z_i,y_i) \exp \left\{ - \frac{h_i^2}{2}\right\}h(u_i|\bm{\nu}) \\
&\propto \frac{\prod_{i=1}^{n} \sqrt{u_i}}{\tau^{n/2}} \exp\left\{ -\frac{1}{2\tau} \sum_{i=1}^{n}u_i\left(z_i - \mu_i \right)^2 \right\} I(z_i,y_i) \exp\left\{-\frac{\sum_{i=1}^{n}h_i^2}{2}\right\}\prod_{i=1}^{n}h(u_i|\bm{\nu}),\nonumber
\end{align}
where $\mu_i = \bm{X}_{i}^\top\bm{\beta}+ \frac{\Delta}{\sqrt{u_i}}(b-h_i)$ and $\bm{\theta} = (\bm{\beta},\delta,\bm{\nu})$. However, since the necessary integrals to obtain the posterior distribution are not easy to calculate, it is not possible to obtain such distribution analytically. However, it is possible to obtain numerical approximation for the marginal posterior distributions of interest by using MCMC algorithms, see \cite{Geman1984} and \cite{Hasting}.

We need to consider a prior distribution for $\bm{\theta}$ such that $\pi(\bm{\theta}) = \pi(\bm{\beta}) \pi(\delta)\pi(\bm{\nu})$. To propose a prior for $\delta$, it is interesting to study the importance of this parameter in the calculation of the skewness of the distributions used in the link functions. For example, consider the Figure \ref{delta_e_gama} that plots the relation between $\delta$ and the Pearson skewness coefficient $\gamma$. We can see that values of $\delta<0.5$ represents a very low skewness coefficient ($\gamma<0.035$) and just for $\delta>0.9$ we have values of $\gamma>0.5$, in special when $\delta>0.99$ we have $\gamma > 0.9173$, for negative $\delta$ the comments are analogous. Then, if we suspect that the link function is skewed we want $\delta$ to have high values in absolute value, and in this case prior distributions with very heavy tails are preferable. 

\begin{figure}[h!]
    \centering
    \includegraphics[scale = 0.6]{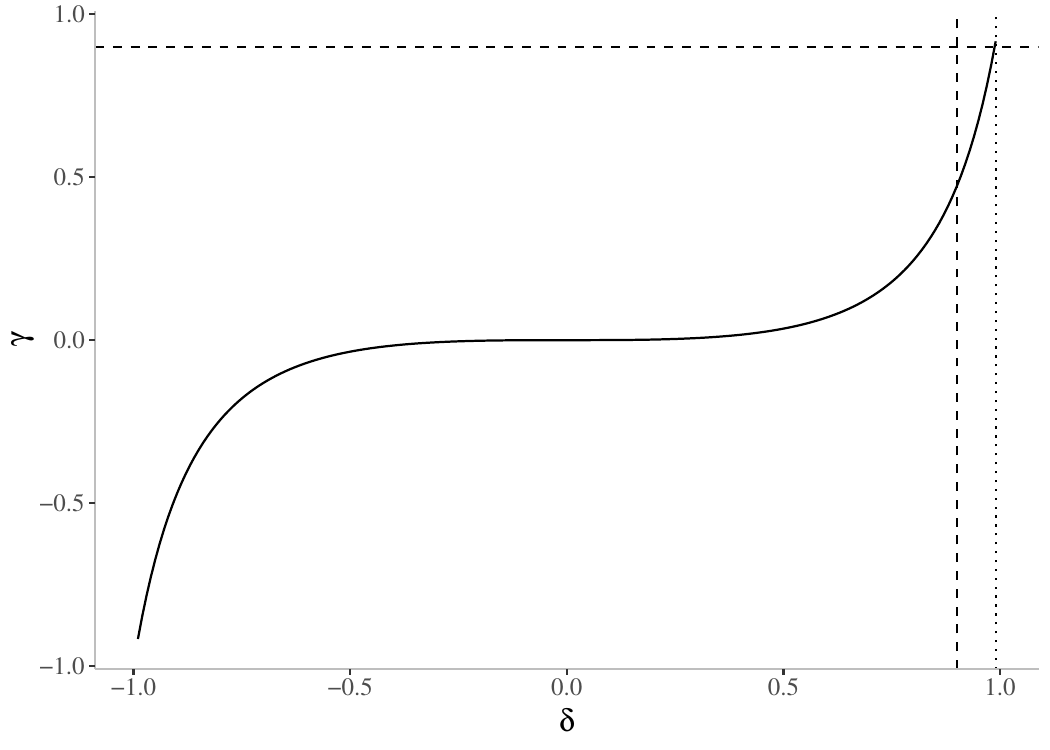}
    \caption{Relation between $\gamma$ and $\delta$. The dashed line represents values of 0.9 and the dotted line represents 0.99.}
    \label{delta_e_gama}
\end{figure}

Said that, based on the prior distribution proposed by \cite{Caio2012} for the skewness parameter of the CSN distribution, we propose to use the following prior $\pi(\delta) = 2/(\pi\sqrt{1-\delta^2})I(\delta \in A)$, where $A$ is $(0,1)$ or $(-1,0)$. In Figure \ref{pi_de_delta} we can see that the density $\pi(\delta)$ have a heavy tail close to 1 when $A = (0,1)$, and close to -1 when $A = (-1,0)$.

\begin{figure}[h!]
     \centering
     \begin{subfigure}[b]{0.49\textwidth}
         \centering
         \includegraphics[width=\textwidth]{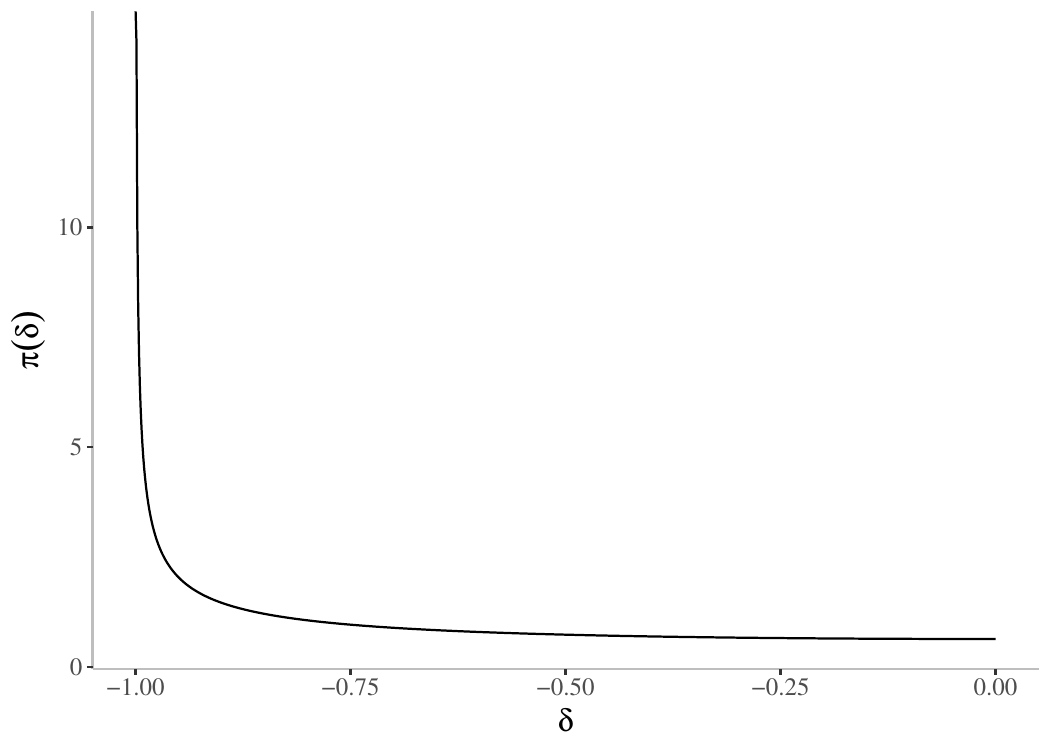}
         \caption{$A = (-1,0)$}
         \label{fig:y equals x}
     \end{subfigure}
     \hfill
     \begin{subfigure}[b]{0.49\textwidth}
         \centering
         \includegraphics[width=\textwidth]{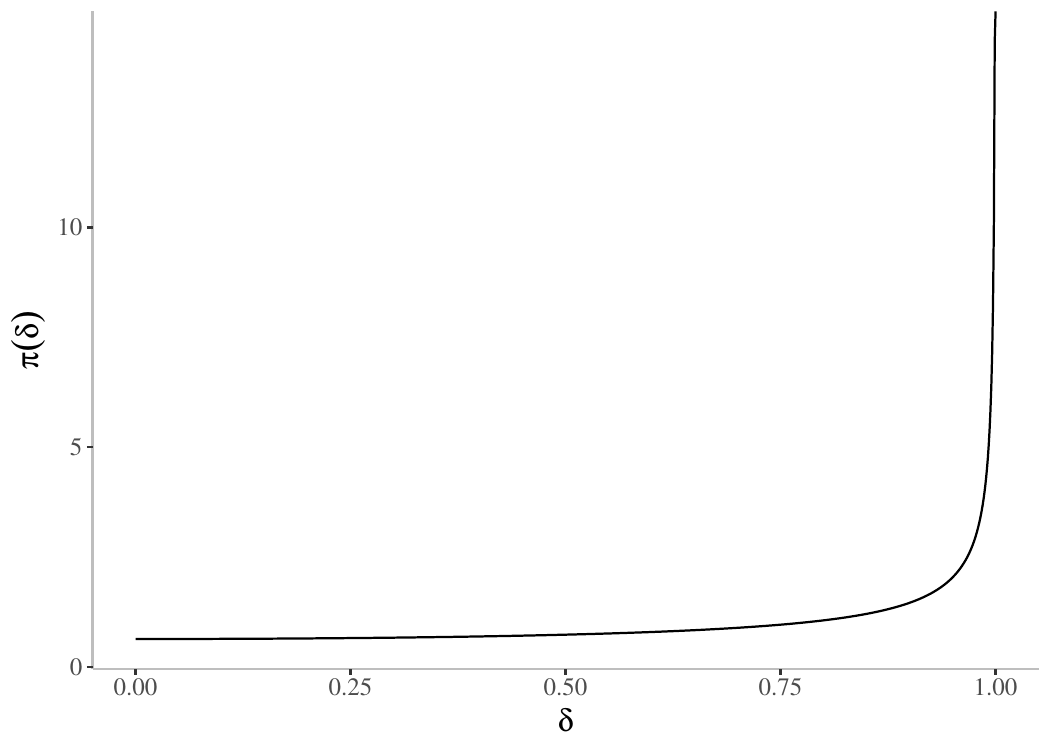}
         \caption{$A = (0,1)$}
         \label{fig:three sin x}
     \end{subfigure}
        \caption{Density function of the prior proposed for $\delta$.}
        \label{pi_de_delta}
\end{figure}

In the next subsections, we performed simulation studies in order to evaluate the performance of the proposed model and the estimation method based on the MCMC algorithms. All these models were implemented in Stan \citep{stan2024} through the interface provided by the \textit{cmdstanr} package \citep{cmdstan} available in R program \citep{R}. The codes are available from the authors upon request. To eliminate the effect of the initial values and to avoid correlations problems, we run a MCMC chain of size 60.000 with a burn-in of 40.000 and a thin of 20, retaining a valid MCMC chain of size 1000.

Since in the CSCN distribution, $U_i$ is a binary random variable and knowing that discrete latent variables are not allowed in the Stan, it is useful to define the model without the presence of $U_i$. One way to do this is to marginalize $U_i$ in the complete likelihood. If we marginalize the complete likelihood in \eqref{completelikelihood} with respect to $U_i$, we will have the distribution of $Z_i|H_i$, for $i=1,\dots,n$. Then,
\begin{align*}
    f_{Z_i|H_i} &= \sum_{u_i \in \{1,\nu_2\}} f_{\bm{Z}_i|H_i,U_i,\bm{\theta}}(z_i) I(z_i,y_i) f(h_i)h(u_i|\bm{\nu})\\
    &= \nu_1f_{\bm{Z}_i|H_i,U_i=\nu_2,\bm{\theta}}(z_i) I(z_i,y_i) f(h_i) + (1-\nu_1)f_{\bm{Z}_i|H_i,U_i=1,\bm{\theta}}(z_i) I(z_i,y_i) f(h_i),
\end{align*}
and the new complete likelihood is then formed by $f_{\bm{Y},\bm{H}, \bm{Z}|\bm{\Theta}}(\bm{y},\bm{h}) = \prod_{i=1}^{n} f_{\bm{Z}_i|H_i,\bm{\theta}}(z_i) I(z_i,y_i) f(h_i)$ which can be used to implement the model based on the CSCN link function in the Stan.

\subsection{Sign choice for $\delta$} \label{sub_sign_delta}

The goal of this simulation study is to define a choice criterion for the sign of $\delta$ when the researcher does not have prior information about the skewness sign in the link function. Then, we propose to use the following scheme:
\begin{enumerate}
    \item Fit a probit regression $p_i = \Phi(\bm{X}_i^\top\bm{\beta})$ to the data, generating $K$ MCMC samples.
    \item For each $k=1,\dots,K$ calculate the residual $\epsilon_{ik} = Z_{ik}-\bm{X}_i^\top\bm{\beta}_k$, where $Z_{ik}$ and $\bm{\beta}_k$ are the $k$th MCMC sample of the posterior of $Z_i$ and $\bm{\beta}$, respectively.
    \item For each $k=1,\dots,K$ calculate the samples skewness of $\epsilon_{1k},\dots,\epsilon_{nk}$ given by $SK(\epsilon_k) = [\sum_{i=1}^n (\epsilon_{ik}-\bar{\epsilon}_k)^3/n]/[\sum_{i=1}^n (\epsilon_{ik}-\bar{\epsilon}_k)^2/n]^{3/2}$, where $\bar{\epsilon}_k = \sum_{i=1}^n\epsilon_{ik}/n$.
    \item Apply some statistic, such as mean, median or mode, to $SK(\epsilon_1),\dots, SK(\epsilon_K)$ to estimate the skewness of the residuals. Then, use the sign of this estimates as the sign of $-\bm{\delta}$.
\end{enumerate}

To evaluate the performance of the proposed scheme to choose the sign of $\delta$, we simulate from samples of size $n=200$, considering 50 replicas and the following stochastic representation:
\begin{align*}
Y_i &= I(Z_i>0),\\
Z_i &= \beta_0+\beta_1 x_i+ \varepsilon_i \quad i=1,\dots,n,
\end{align*}
where $\beta_0=1$, $\beta_1=2$, $\varepsilon_i$ belongs to the CSN distribution with $\sigma^2=1$ and $\delta \in \{-0.99,-0.9,0.9,0.99\}$, which allows the model to have strong negative, medium negative, medium positive and strong positive skewness, respectively. The covariate was simulated from a N(0,1) and was standardized. We adopted weakly informative priors for the regression coefficients: $\beta_0 \sim N(0,1000)$ and $\beta_1 \sim N(0,1000)$.

We can see from Table \ref{signal_choice} that under strong skewness the scheme selects the correct sign 94\% of the time, while under medium skewness it selects 72\% of the time correctly, showing that the proposed scheme works well if we suspect that there is a considerable skewness in the link function.

\begin{table}[h!]
    \centering
    \caption{Number of times that the correct sign of $\delta$ was selected correctly under posterior mean, median and mode estimates.}
    \label{signal_choice}
    \begin{tabular}{c|ccc}
     \hline
     $\delta$ & Posterior mean & Posterior median & Posterior mode \\
     \hline
     0.99 & 47 & 47 & 43\\
     0.9 & 36 & 37 & 39\\
     -0.9 & 36 & 36 & 36 \\
     -0.99 & 47 & 47 & 41\\
     \hline
    \end{tabular}
\end{table}

\subsection{Prior for regression coefficients} \label{ss3: reg recov}

In parameter recovery simulation studies, an overestimation of the regression coefficients was initially noticed when we assume that $\bm{\beta} \sim N (\bm{\mu}_b, \bm{\Sigma}_b)$, especially when considering medium and small sample sizes, resulting in the need to define priors for $\bm{\beta}$ in order to reduce the bias. Then, we propose to use a Zellner's g-prior~\citep{zellner1986assessing} given by $\bm{\beta} \sim N (\bm{\mu}_b, g\bm{\Sigma}_b)$, which allows a prior correlation between the regression coefficients and automatically selects the prior variance~\citep{gosho2023bias}. \cite{liang2008mixtures} propose to use a hyper-g prior given by $\pi(g|\alpha) = \frac{a-2}{4}(1+g)^{-\alpha/2}$, $g > 0$, $\alpha>2$, which allows us to control the degree of shrinkage noticing that the shrinkage factor $g/(1+g) \sim beta(1,\alpha/2-1)$. Then, if $\alpha=4$ the shrinkage factor is uniform and if $\alpha>4$ there is more mass on shrinkage values near to 0~\citep{liang2008mixtures}. Said that, we will assume two possibilities: $\alpha=4$ or $2<\alpha\leq 4$ with $\alpha \sim Uniform(2,4)$. Also, we will assume $\bm{\mu}_b = \bm{0}$ which is a noninformative choice and $\bm{\Sigma} = \text{diag}(1/2,\dots,1/2)$ as recommended by \cite{held2017adaptive}.

We simulate from a sample of size $n = 100$, considering 10 replicas and the same model of the subsection \ref{sub_sign_delta}. In addition to the two hyper-g priors commented, we also consider classical normal priors $\beta_0 \sim N(0,1000)$ and $\beta_1 \sim N(0,1000)$ for comparative purposes. From Table \ref{reg_hyper} it can be seen that the classic normal prior causes an overestimation, and the use of the hyper-g priors reduces the bias. More specifically, if we use the hyper-g prior with $\alpha=4$ we noticed a reduction of 90.45\% in the bias for $\beta_0$ and 84.89\% for $\beta_1$. Comparing the hyper-g priors with $\alpha=1$ and $2<\alpha\leq 4$ we do not notice any major differences, then, we propose the use of hyper-g priors with $\alpha=4$ for the SMCSN link models proposed in this paper.

\begin{table}[h!]
    \centering
        \caption{Estimates for regression parameters using the classic normal prior and the hyper-g prior considering $\alpha=4$ and $2 < \alpha \leq 4$, under the posterior mean, median and mode statistics.}
    \label{reg_hyper}
    \begin{tabular}{c|cccc}
     \hline
     Parameter & prior & Posterior mean & Posterior median & Posterior mode \\
     \hline
     \multirow{ 3}{*}{$\beta_0$} & normal & 1.20360 & 1.18169 & 1.12366\\
     & hyper-g, $\alpha=4$ & 1.08035 & 1.07125 & 0.98819 \\
     & hyper-g, $2<\alpha \leq 4$ & 1.10122 & 1.08472 & 1.08610 \\
     \hline
     \multirow{ 3}{*}{$\beta_0$} & normal & 2.34638 & 2.31577 & 2.27255\\
     & hyper-g, $\alpha=4$ & 2.10168 & 2.08638 & 2.04118 \\
     & hyper-g, $2<\alpha \leq 4$ & 2.14862 & 2.11199 & 1.96291 \\
     \hline
    \end{tabular}
\end{table}

\subsection{Parameter recovery simulation study} \label{ss3: par recov}
The main goal of this simulation study is to measure the impact of the sample size on the parameter recovery. We have considered different scenarios based on the crossing of the levels of some factors of interest. For the five regression models explored in this work, we simulate from samples of size n=100 and 250, considering 10 replicas. We have generated replicas from
\begin{align*}
Y_i &= I(Z_i>0),\\
Z_i &= \beta_0+\beta_1 x_i+ \varepsilon_i \quad i=1,\dots,n
\end{align*}
where $\beta_0=1$, $\beta_1=2$, $\varepsilon_i$ belongs to the SMCSN family of distributions with $\sigma^2=1$ and $\delta \in \{-0.99,0.99\}$, which allows the model to have strong negative and strong positive skewness, respectively. Also, we set $\nu = 3$ for the CST and CSS distribution and $\bm{\nu} = (\nu_1,\nu_2) = (0.7,0.7)$ for CSCN distribution.
These values for $\bm{\nu}$ were chosen in order to have distributions with heavy tails. The covariate was simulated from a N(0,1) distributions and was standardized.

Let $\theta$ be any parameter to be estimated and $\hat{\theta}_r$ the estimate based on some posterior statistic (mean, median or mode) from the $r$-th replica. We compare the performance of the estimators using some appropriate statistics: mean of the estimates of $\theta$ (Est) $\bar{\hat{\theta}} =\frac{\sum_{r=1}^{10} \hat{\theta_r} }{10}$, standard deviation of the estimates (SD) $SD_{\theta} =\sqrt{\frac{\sum_{r=1}^{10} ( \hat{\theta_r}  -\bar{\hat{\theta}} )^2  }{9}}$, relative bias (Rel Bias) $\frac{\rvert  Bias_{\theta}\rvert }{\theta}$, where $Bias_{\theta} = \bar{\hat{\theta}} -\theta$ is the bias and the mean square error (MSE) $MSE_{\theta} = Bias_{\theta}^2 + SD_{\theta}^2$.

For all distributions we compare the posterior statistics using the MSE. For the regression parameters ($\beta_0$,$\beta_1$) the posterior median provides the best estimates for CSN, CST and CSS link functions, furthermore, the mean is the best for the CSCN link function. On the other hand, for $\delta$, the posterior mode has the best performances for all models.  Finally, for $\bm{\nu}$ the posterior mode is suggested for the CST link function, the posterior median for the CSS link function, and the posterior mean for the CSCN link function.

The results only for the selected posterior statistics for all models are showed in Tables \ref{recor:csn} to \ref{recor:cscn}. From these tables we can notice that for all parameters, all models and all sample sizes, the estimates are accurate with low bias. For the regression parameters, it is noticed that the MSE decreases with the increase in sample size, this happens also for $\delta$, except for the CST link model in which an underestimation can be seen. Regarding the shape parameters ($\bm{\nu}$), it is noticed that there was no improvement when we increasing the sample size, in general a low bias is noticed and for the CST and CSS links we can see a higher MSE compared to the other parameters, possibly caused by the standard deviation, however, all these characteristics did not affect the estimation of the regression and skewness parameters.

\begin{table}[h!]
\centering
	\caption{Results of the simulation study for the CSN link model.}
	\begin{tabular}{l|cccccc}
			\hline
			n	&  Parameter & Real & Est & SD & Rel Bias & MSE\\ 
			\hline
			\multirow{4}{*}{100}
			& $\beta_0$  & 1.0000  & 1.0196 & 0.2432 & 0.0196 & 0.0595  \\
			& $\beta_1$  & 2.0000  & 2.0123 & 0.3036 & 0.0061 & 0.0923 \\
			& $\delta$   & 0.99 & 0.9618 & 0.0485 & 0.0284 & 0.1013 \\
			\hline	
			\multirow{4}{*}{250}
			& $\beta_0$  & 1.0000  & 1.0410 & 0.2136 & 0.0410 & 0.0473 \\
			& $\beta_1$  & 2.0000  & 1.9861 & 0.1972 & 0.0069 & 0.0391\\
			& $\delta$   & 0.99 & 0.9763 & 0.0323 & 0.0138 & 0.0012 \\
			\hline
	\end{tabular}
	\label{recor:csn}
\end{table}

\begin{table}[h!]
\centering
	\caption{Results of the simulation study for the CST link model.}
	\begin{tabular}{l|cccccc}
			\hline
			n	&  Parameter & Real & Est & SD & Rel Bias & MSE\\ 
			\hline
			\multirow{4}{*}{100}
			& $\beta_0$  & 1.0000  & 0.9349 & 0.2395 & 0.0651 & 0.0616  \\
			& $\beta_1$  & 2.0000  & 1.9188 & 0.7255 & 0.0406 & 0.5330 \\
			& $\delta$   & 0.99   & 0.9752 & 0.0217 & 0.0150 & 0.0007 \\
                & $\nu$      & 3       & 2.9156 & 0.8517 & 0.0281 & 0.7326 \\
			\hline	
			\multirow{4}{*}{250}
			& $\beta_0$  & 1.0000  & 0.9469 & 0.2003 & 0.0531 & 0.0429 \\
			& $\beta_1$  & 2.0000  & 1.9140 & 0.1809 & 0.0430 & 0.0401\\
			& $\delta$   & 0.99   & 0.9368 & 0.0865 & 0.0538 & 0.0103 \\
                & $\nu$      & 3       & 3.0892 & 1.0091 & 0.0297 & 1.0262 \\
			\hline
	\end{tabular}
	\label{recor:cst}
\end{table}

\begin{table}[h!]
\centering
	\caption{Results of the simulation study for the CSS link model.}
	\begin{tabular}{l|cccccc}
			\hline
			n	&  Parameter & Real & Est & SD & Rel Bias & MSE\\ 
			\hline
			\multirow{4}{*}{100}
			& $\beta_0$  & 1.0000  & 1.0482 & 0.4803 & 0.0482 & 0.2330  \\
			& $\beta_1$  & 2.0000  & 1.8997 & 0.5495 & 0.0501 & 0.3120 \\
			& $\delta$   & 0.99   & 0.9520 & 0.0970 & 0.0384 & 0.0108 \\
                & $\nu$      & 3       & 3.1564 & 0.5371 & 0.0521 & 0.3129 \\
			\hline	
			\multirow{4}{*}{250}
			& $\beta_0$  & 1.0000  & 0.9916 & 0.1849 & 0.0083 & 0.0343 \\
			& $\beta_1$  & 2.0000  & 1.9526 & 0.2457 & 0.0237 & 0.0626\\
			& $\delta$   & 0.99   & 0.9728 & 0.0359 & 0.0174 & 0.0016 \\
                & $\nu$      & 3       & 3.0825 & 0.6419 & 0.0275 & 0.4189 \\
			\hline
	\end{tabular}
	\label{recor:css}
\end{table}

\begin{table}[h!]
\centering
	\caption{Results of the simulation study for the CSCN link model.}
	\begin{tabular}{l|cccccc}
			\hline
			n	&  Parameter & Real & Est & SD & Rel Bias & MSE\\ 
			\hline
			\multirow{4}{*}{100}
			& $\beta_0$  & 1.0000  & 1.1020 & 0.2644 & 0.1020 & 0.0803  \\
			& $\beta_1$  & 2.0000  & 2.0120 & 0.3845 & 0.0060 & 0.1480 \\
			& $\delta$   & 0.99   & 0.9764 & 0.0206 & 0.0137 & 0.0006 \\
                & $\nu_1$    & 0.5     & 0.4567 & 0.0209 & 0.0866 & 0.0023 \\
                & $\nu_2$    & 0.5     & 0.5486 & 0.0352 & 0.0973 & 0.0036 \\
			\hline	
			\multirow{4}{*}{250}
			& $\beta_0$  & 1.0000  & 1.0101 & 0.1240 & 0.0101 & 0.0155 \\
			& $\beta_1$  & 2.0000  & 2.0006 & 0.1693 & 0.0003 & 0.0287\\
			& $\delta$   & 0.99   & 0.9786 & 0.0159 & 0.0115 & 0.0004 \\
                & $\nu_1$    & 0.5     & 0.4689 & 0.0185 & 0.0622 & 0.0013 \\
                & $\nu_2$    & 0.5     & 0.5829 & 0.0392 & 0.1658 & 0.0084 \\
			\hline
	\end{tabular}
	\label{recor:cscn}
\end{table}

\section{Residual analysis} \label{sss3: res}

Similarly to the latent Bayesian residual for the skew probit regression developed in \cite{farias2012}, we can define the latent residuals for the binary regression model with link function based on the SMCSN distributions from the stochastic representation given in \eqref{prop3}. Then, we can define the residual for the $i$th subject as
\begin{align}
\epsilon_i = \frac{Z_i - \Xb^\top_i\betab + U_i^{-1/2}\Delta(H_i-b) }{\sqrt{\tau}},
\label{res_bin}
\end{align}
where $Z_i, U_i$ and $H_i$ are the latent variables. It follows that, conditioned on $(\betab,\delta)$, the residual \eqref{res_bin} is normally distributed a priori. A way to check lack of fit is to build the normal envelope plot for these residuals. 

\section{Application} \label{s3: application}

We analyze, using the developed models and the probit one, the heart disease data~\citep{detrano1989international} available in the UCI machine learning repository at \url{https://archive.ics.uci.edu/dataset/45/heart+disease}. This data consists of a sample of 303 heart disease diagnosis from patients of Cleveland. The goal of this example is to relate the presence of heart disease in the patient with some covariates of interest, also we will show that the class of models proposed in this paper can provide better fits than the usual probit model. In this paper we only considered subjects with no missing values, generating 297 subjects. We considered as the binary response $Y_i$ the presence of heart disease with values 1 for presence and 0 for absence. Also, we considered the covariates selected by \cite{lee2019identifiability} which implies the following model:
\begin{align*}
Y_i &= I(Z_i>0),\\
Z_{i} &= \beta_0+\beta_1 sex_{i}+ \sum^3_{j=1}\beta_{2j} CP_{i} + \beta_3 BP_{i}+  \sum^2_{j=1}\beta_{4j} slope_{i} + \beta_5 CF_{i} + \sum^2_{j=1}\beta_{6j} thal_{i} + \varepsilon_i,
\end{align*}
for $i=1,\dots,297$, where $sex_{i}$ is the sex (1=male, 0=female), 
$CP_i$ represents the chest pain types (0 = typical angina, 1 = atypical angina, 2 = non-anginal pain, 3 = asymptomatic), $slope_{i}$ is the slope of the peak exercise ST segment (0 = upsloping, 1 = flat, 2 = downsloping), $thal_{i}$ is the thallium heart scan results (0 = normal, 1 = fixed defect, 2 = reversable defect), $BP_i$ is the resting blood pressure on admission to the hospital (mmHg) and $CF_{i}$ is the number of major vessels (0,1,2 or 3) colored by flourosopy, centered in their respective mean and standard deviation. 

We fitted five models, assuming that: $\varepsilon_i \stackrel{iid}{\sim} CST(0,1,-\delta,\nu)$, or , $\varepsilon_i \stackrel{iid}{\sim} CSS(0,1,-\delta,\nu)$, or $\varepsilon_i \stackrel{iid}{\sim} CSCN(0,1,-\delta,\nu_1,\nu_2)$, or $\varepsilon_i \stackrel{iid}{\sim} CSN(0,1,-\delta)$, or $\varepsilon_i \stackrel{iid}{\sim} N(0,1)$, that we denote, respectively, by CST, CSS, CSCN, CSN and N. 
The parameters for the MCMC algorithm and the adopted prior distributions were the same used in the simulation study described in Section \ref{ss3: par recov}. From Figure \ref{qqplot}, QQ plots with envelopes for all fitted models are shown. It is possible to see that for all models there are some points lying outside the confidence bands, possibly caused by the influential observations and/or lack of skewness, except for the CSS model.

\begin{figure}[h!]
     \centering
     \begin{subfigure}[b]{0.49\textwidth}
         \centering
         \includegraphics[width=\textwidth]{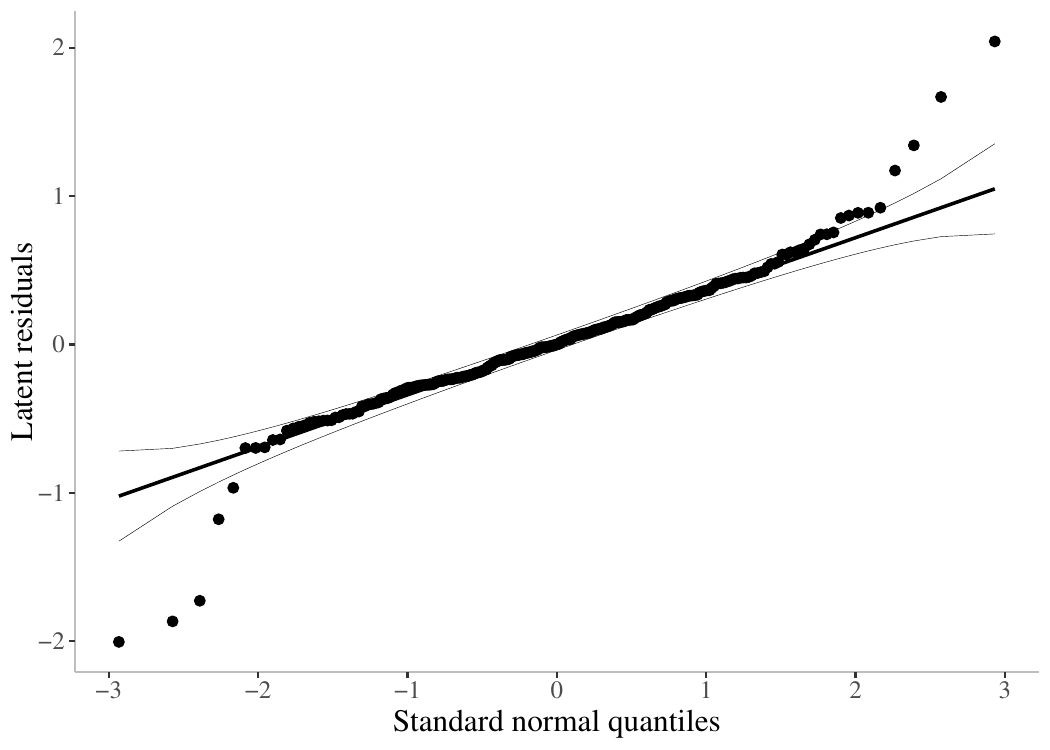}
         \caption{Probit}
     \end{subfigure}
     \begin{subfigure}[b]{0.49\textwidth}
         \centering
         \includegraphics[width=\textwidth]{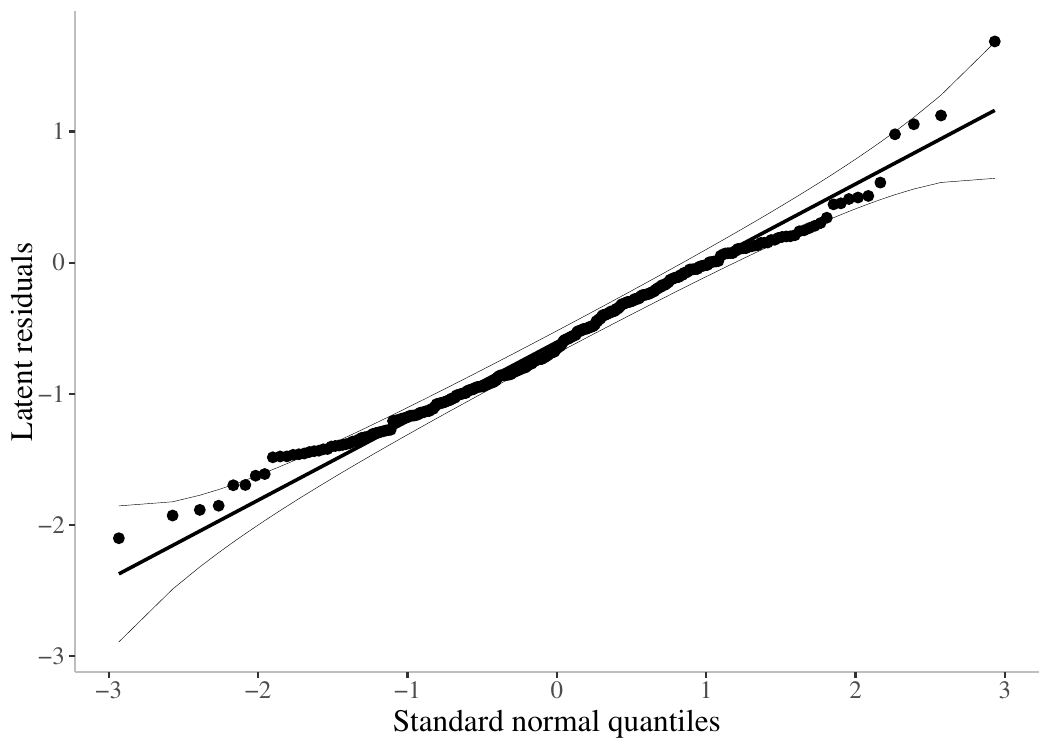}
         \caption{CSN}
     \end{subfigure}
     \begin{subfigure}[b]{0.49\textwidth}
         \centering
         \includegraphics[width=\textwidth]{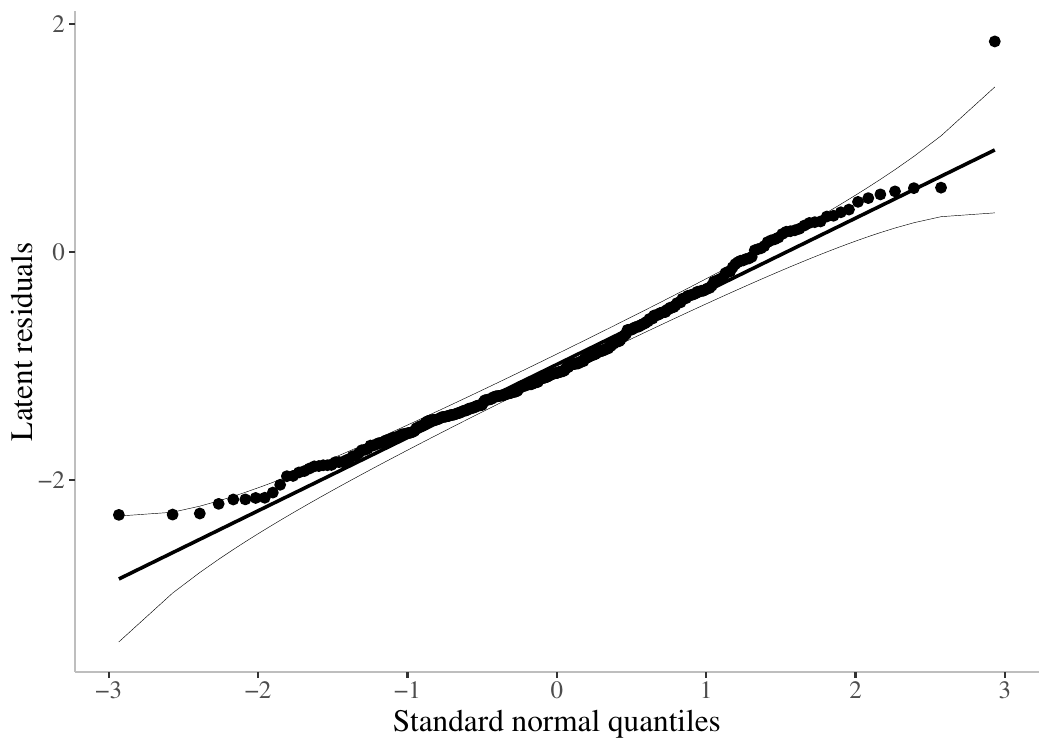}
         \caption{CST}
     \end{subfigure}
     \begin{subfigure}[b]{0.49\textwidth}
         \centering
         \includegraphics[width=\textwidth]{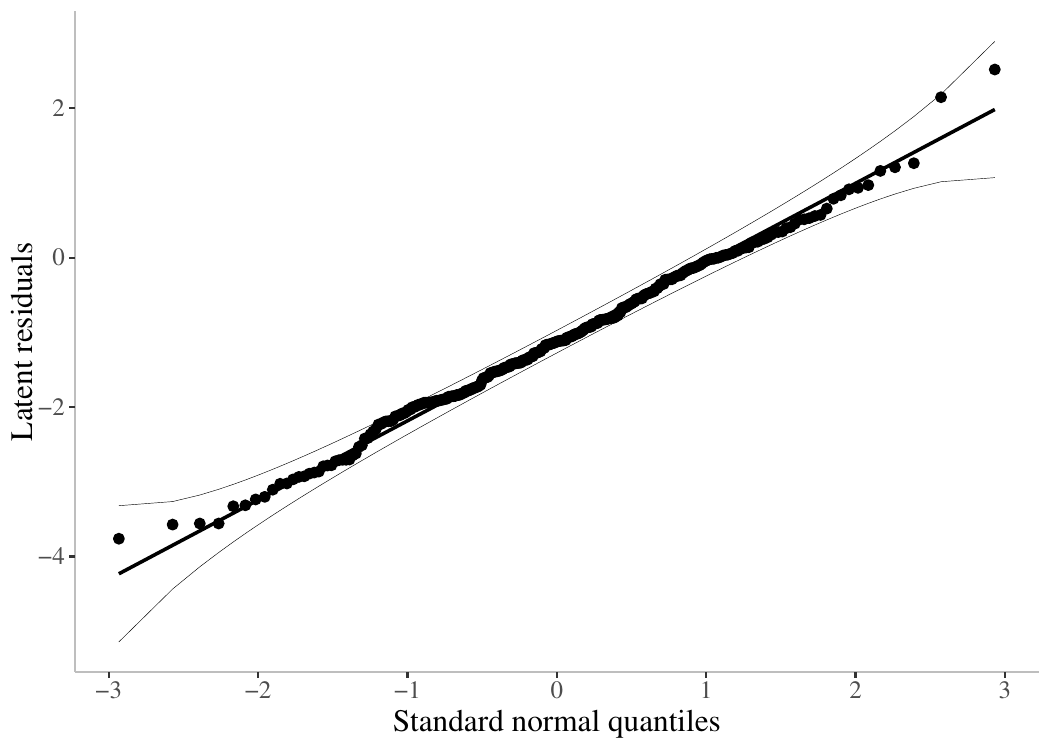}
         \caption{CSS}
     \end{subfigure}
     \begin{subfigure}[b]{0.49\textwidth}
         \centering
         \includegraphics[width=\textwidth]{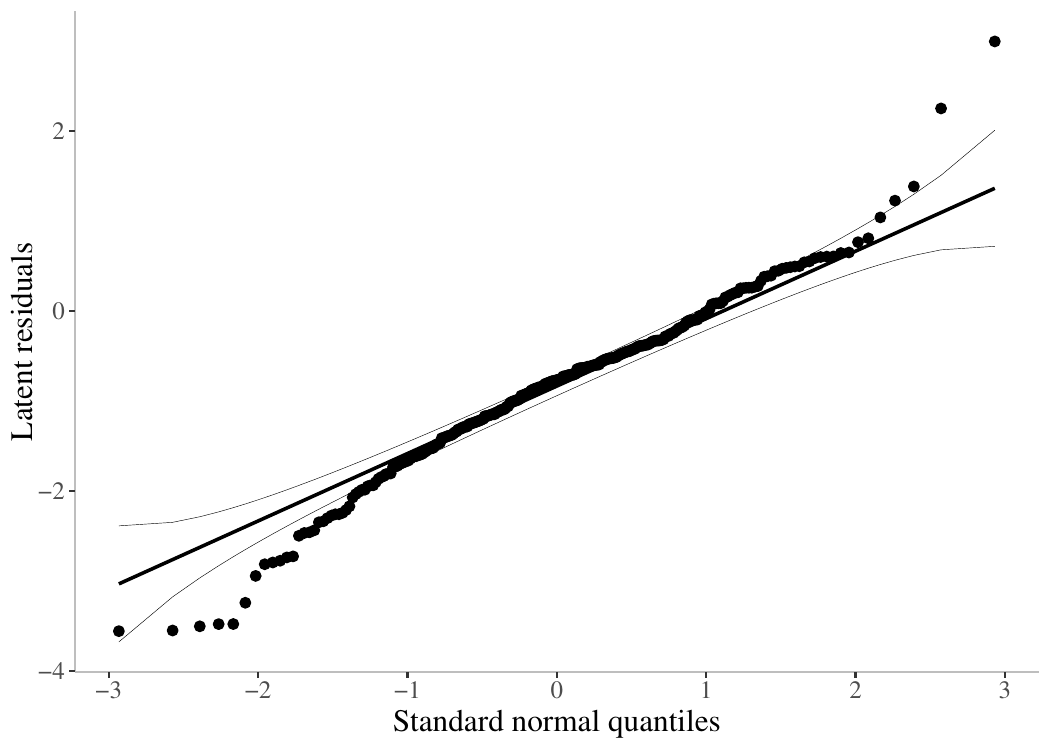}
         \caption{CSCN}
     \end{subfigure}
        \caption{QQ plots with envelopes for the fitted models.}
        \label{qqplot}
\end{figure}

From Figure \ref{posterior_delta}, the posterior distributions indicate that the the link function is positive skewed, since $\delta$ is concentrated towards positive values, then a skewed link function is more appropriate for this data. Figure \ref{posterior_nu} presents the posterior distributions for the shape parameters, we can see that the point of mass is concentrated in values that induce heavy tails in the distribution of $F(\bm{X}_i^\top \bm{\beta}|\delta,\bm{\nu})$, then we have an indication that CST, CSS and CSCN overperform the CSN, specially CSS that have the best QQ plot. Said that, the most appropriate model for the data modeling is the CSS model.

\begin{figure}[h!]
     \centering
     \begin{subfigure}[b]{0.49\textwidth}
         \centering
         \includegraphics[width=\textwidth]{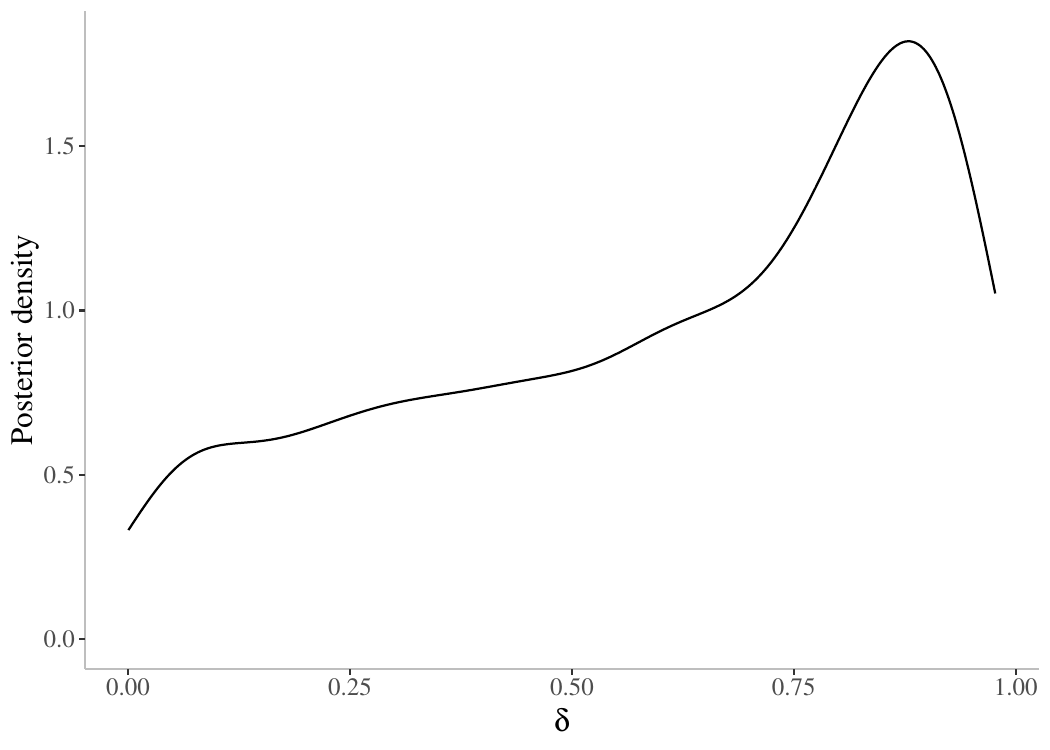}
         \caption{CSN}
     \end{subfigure}
     \begin{subfigure}[b]{0.49\textwidth}
         \centering
         \includegraphics[width=\textwidth]{delta_csn.pdf}
         \caption{CST}
     \end{subfigure}
     \begin{subfigure}[b]{0.49\textwidth}
         \centering
         \includegraphics[width=\textwidth]{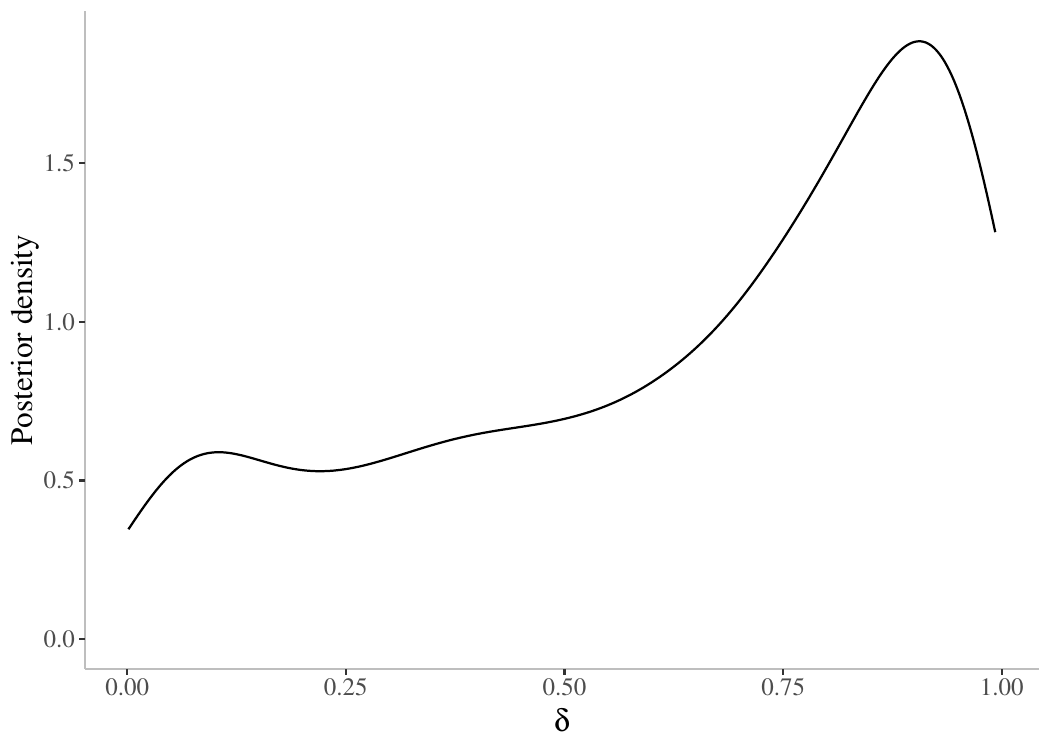}
         \caption{CSS}
     \end{subfigure}
     \begin{subfigure}[b]{0.49\textwidth}
         \centering
         \includegraphics[width=\textwidth]{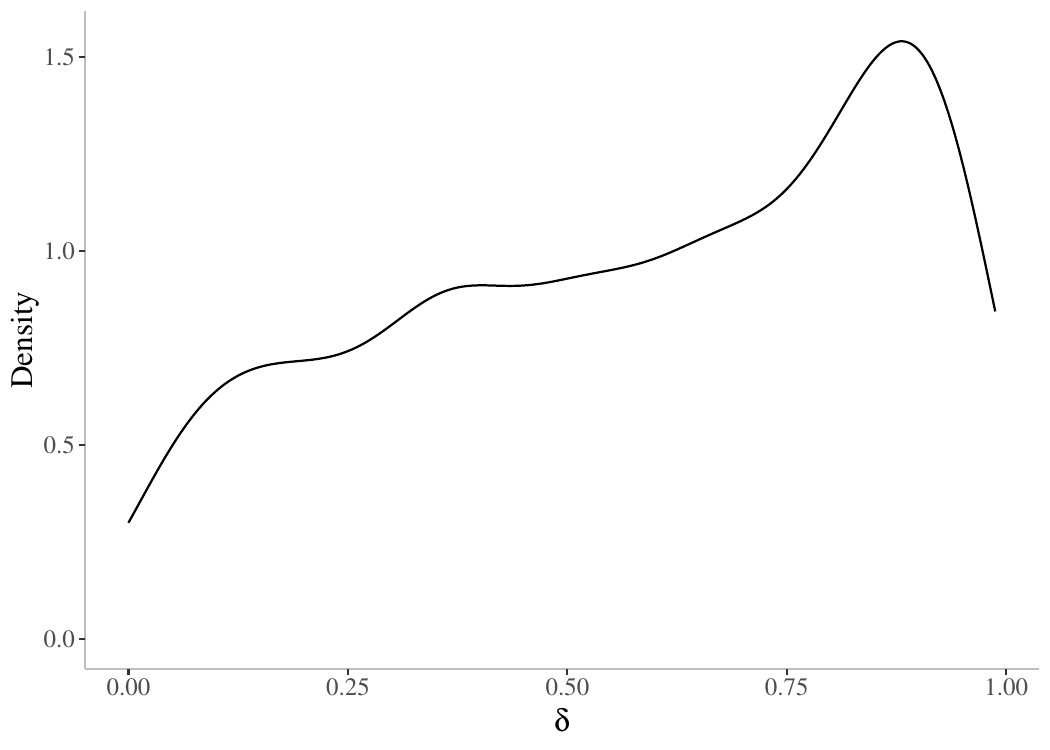}
         \caption{CSCN}
     \end{subfigure}
     \hfill
        \caption{Posterior distribution of $\delta$ for the CSN, CST, CSS and CSCN models.}
        \label{posterior_delta}
\end{figure}

\begin{figure}[h!]
     \centering
     \begin{subfigure}[b]{0.49\textwidth}
         \centering
         \includegraphics[width=\textwidth]{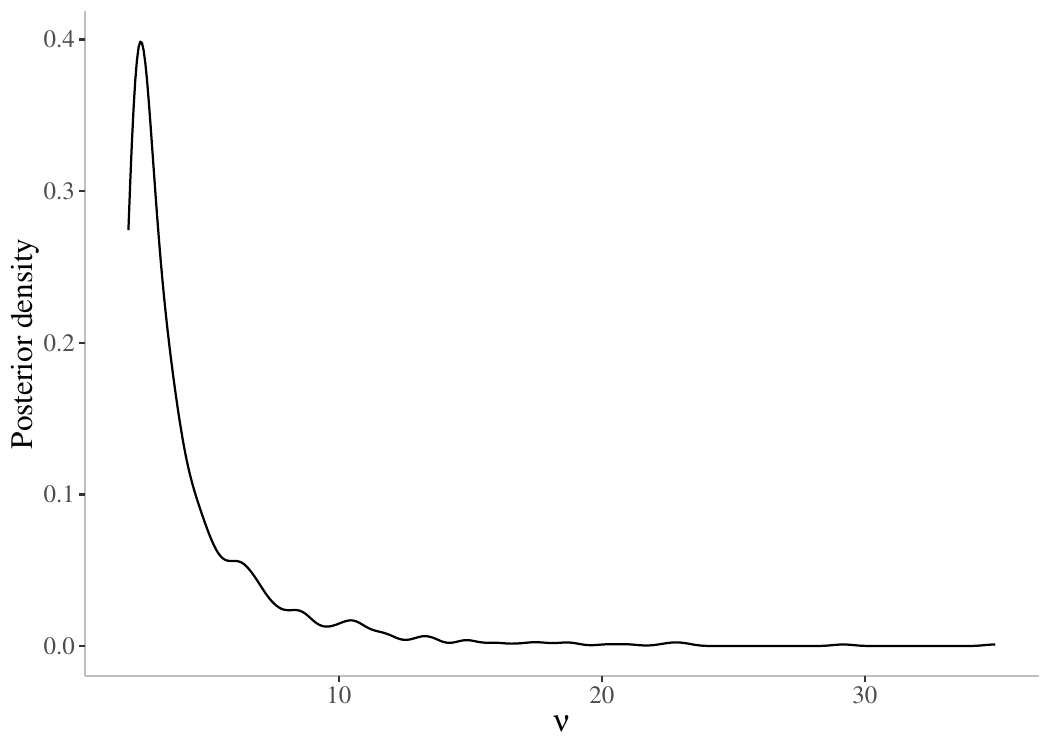}
         \caption{CST}
     \end{subfigure}
     \begin{subfigure}[b]{0.49\textwidth}
         \centering
         \includegraphics[width=\textwidth]{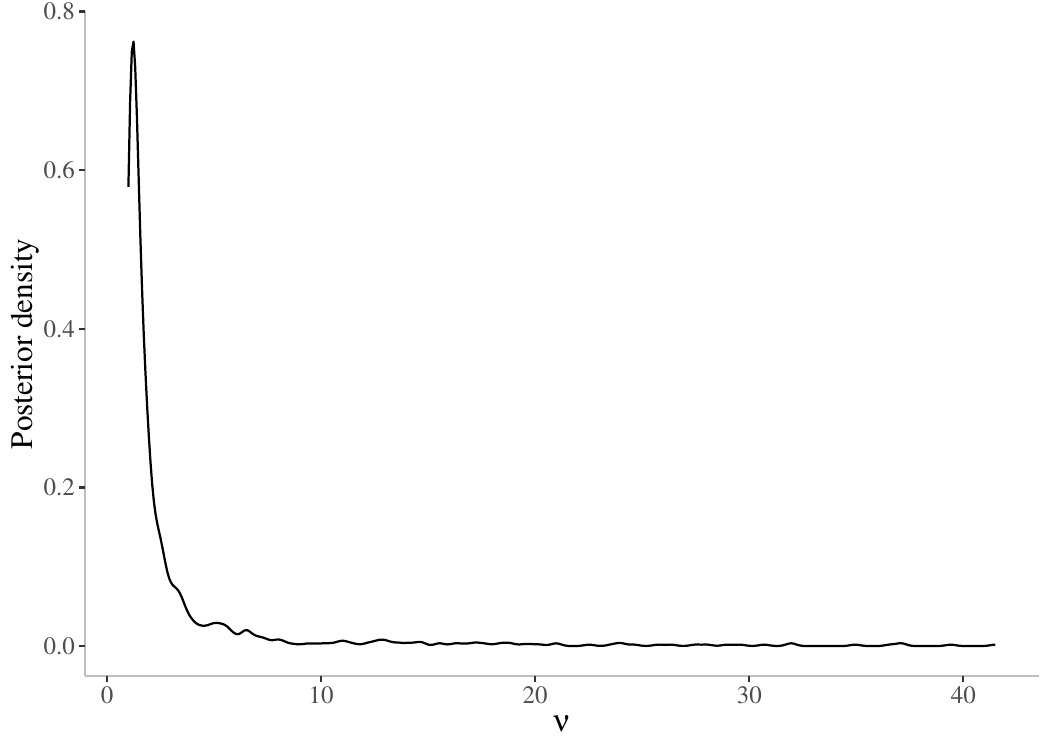}
         \caption{CSS}
     \end{subfigure}
     \begin{subfigure}[b]{0.49\textwidth}
         \centering
         \includegraphics[width=\textwidth]{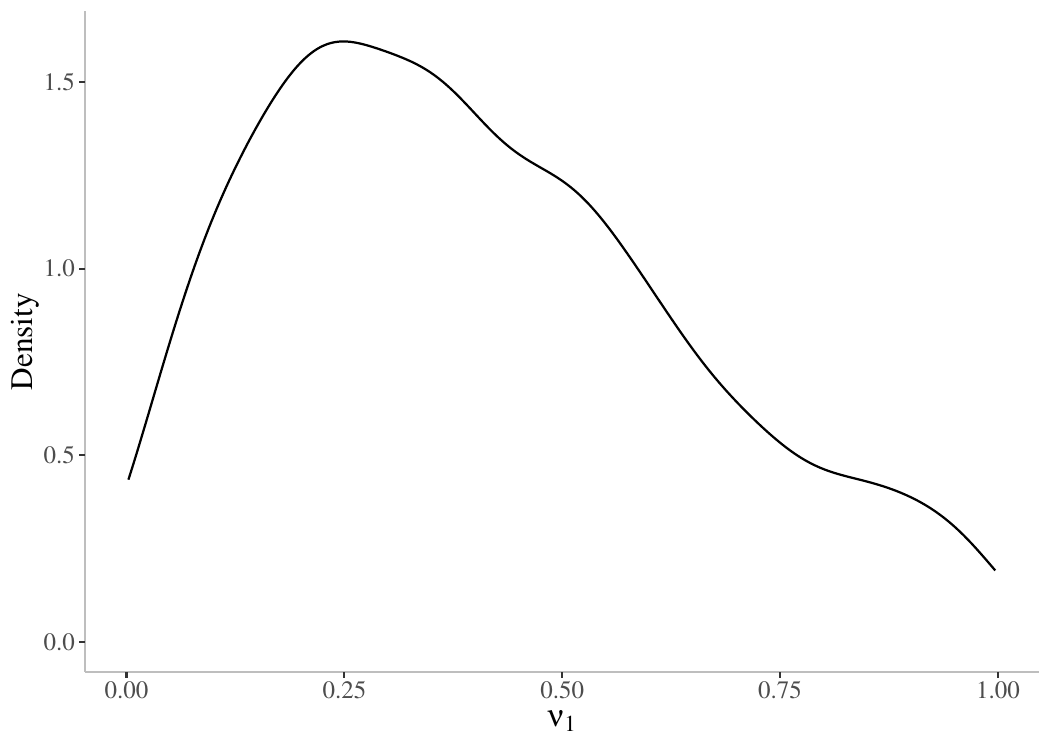}
         \caption{CSCN ($\nu_1$)}
     \end{subfigure}
     \begin{subfigure}[b]{0.49\textwidth}
         \centering
         \includegraphics[width=\textwidth]{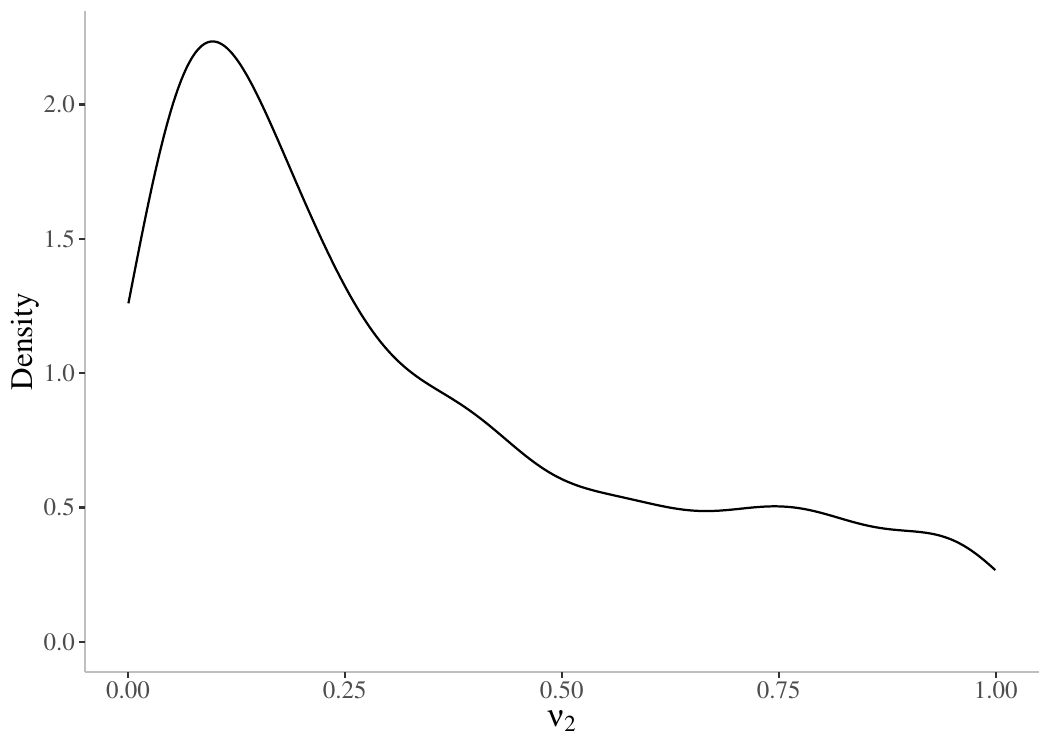}
         \caption{CSCN ($\nu_2$)}
     \end{subfigure}
        \caption{Posterior distribution of $\nu$ for CST and CSS and $\nu_1$ and $\nu_2$ for the CSCN}
        \label{posterior_nu}
\end{figure}

In Table \ref{posterior_estimates} we have the estimates of all fitted models. For all models, we can see the same signs for $\bm{\beta}$. In particular, if we compare the CSS that presented the best QQ plot with the CSN we see that the intercept has a greater magnitude in the CSS, the effect of $\beta_{21}$ is almost null in the CSN while in the CSS it is not, and the skewness parameter value in the CSS is bigger than in the CSN. Interpreting the parameters we have that male has more chance of heart disease than female, asymptomatic patients has more chance of heart disease, resting blood pressure on admission to the hospital has a positive effect on the heart disease chance, the slope of the peak exercise ST segment has more chance of heart disease, the number of major vessels has a positive effect on the heart disease, and the patients with thallium heart scan results reversable defect has more chance of heart disease.

\begin{table}[h!]
	\centering
	\caption{Posterior parameter estimates and HPD for the parameters of the fitted models.}
 \label{posterior_estimates}
	\resizebox{\textwidth}{!}{\begin{tabular}{lcccccccccccccc}
	\hline
        Model & $\beta_0$ & $\beta_1$ & $\beta_{21}$ & $\beta_{22}$ & $\beta_{23}$ & $\beta_3$ & $\beta_{41}$ & $\beta_{42}$ & $\beta_5$ & $\beta_{61}$ & $\beta_{62}$ & $\delta$ & $\nu (\nu_1)$ & $\nu_2$\\
        \hline
        Probit & -1.60 & 0.61 & $<0.01$ & -0.28 & 0.92 & 0.22 & 0.79 & 0.50 & 0.61 & 0.13 & 0.82 & - & - & -\\
        & (-2.28,-0.96) & (0.16,1.03) & (-0.65,0.68) & (-0.83,0.33) & (0.34,1.49) & (0.04,0.42) & (0.40,1.21) & (-0.23,1.17) & (0.39,0.82) & (-0.55,0.84) & (0.37,1.21) & - & - & - \\
        CSN & -1.56 & 0.60 & $<0.01$ & -0.26 & 0.91 & 0.22 & 0.77 & 0.45 & 0.63 & 0.11 & 0.81 & 0.90 & - & -\\
        & (-2.25,-0.93) & (0.19,1.05) & (-0.56,0.70) & (-0.80,0.33) & (0.40,1.49) & (0.02,0.41) & (0.39,1.19) & (-0.18,1.19) & (0.42,0.85) & (-0.56,0.84) & (0.43,1.27) & (0.07,0.97) & - & -\\
        CST & -2.16 & 0.77 & 0.16 & -0.35 & 1.27 & 0.28 & 1.17 & 0.71 & 0.96 & 0.14 & 1.16 & 0.91 & 2.08 & - \\
        & (-3.43,-1.02) & (0.14,1.45) & (-0.91,1.08) & (-1.10,0.55) & (0.47,2.17) & (0.03,0.53) & (0.50,1.93) & (-0.18,1.67) & (0.55,1.50) & (-0.83,1.07) & (0.61,1.89) & (0.08,0.97) & (2.00,10.35)\\
        CSS & -2.26 & 0.82 & 0.11 & -0.38 & 1.32 & 0.29 & 1.19 & 0.76 & 1.00 & 0.12 & 1.21 & 0.98 & 1.09 \\
        & (-3.74,-0.98) & (0.18,1.58) & (-0.91,1.13) & (-1.30,0.46) & (0.48,2.47) & (-0.02,0.60) & (0.51,2.18) & (-0.25,1.92) & (0.53,1.64) & (-0.91,1.17) & (0.53,2.03) & (0.07,0.99) & (1.00,12.75)\\
        CSCN & -2.17 & 0.80 & 0.09 & -0.38 & 1.31 & 0.29 & 1.16 & 0.73 & 0.97 & 0.17 & 1.17 & 0.89 & 0.22 & 0.06\\
        & (-4.88,-0.96) & (0.09,1.94) & (-0.96,1.39) & (-1.40,0.53) & (0.30,2.95) & (0,0.64) & (0.40,2.69) & (-0.15,2.33) & (0.44,2.10) & (-1.04,1.33) & (0.42,2.59) & (0.09,0.99) & (0,0.87) & (0,0.91)\\
        \hline
	\end{tabular}}
\end{table}

In the Figure \ref{prob_success} we have the probability of heart disease as a function of $\eta_i$ for all models, we can see that the curves of probability of distributions with scale mixtures have different behaviour in comparison with CSN, which can cause different interpretations among the fitted models, also, for CST, CSS and CSCN link functions we can see a high probability rate for $\eta_i > 0$.

\begin{figure}[h!]
     \centering
     \begin{subfigure}[b]{0.49\textwidth}
         \centering
         \includegraphics[width=\textwidth]{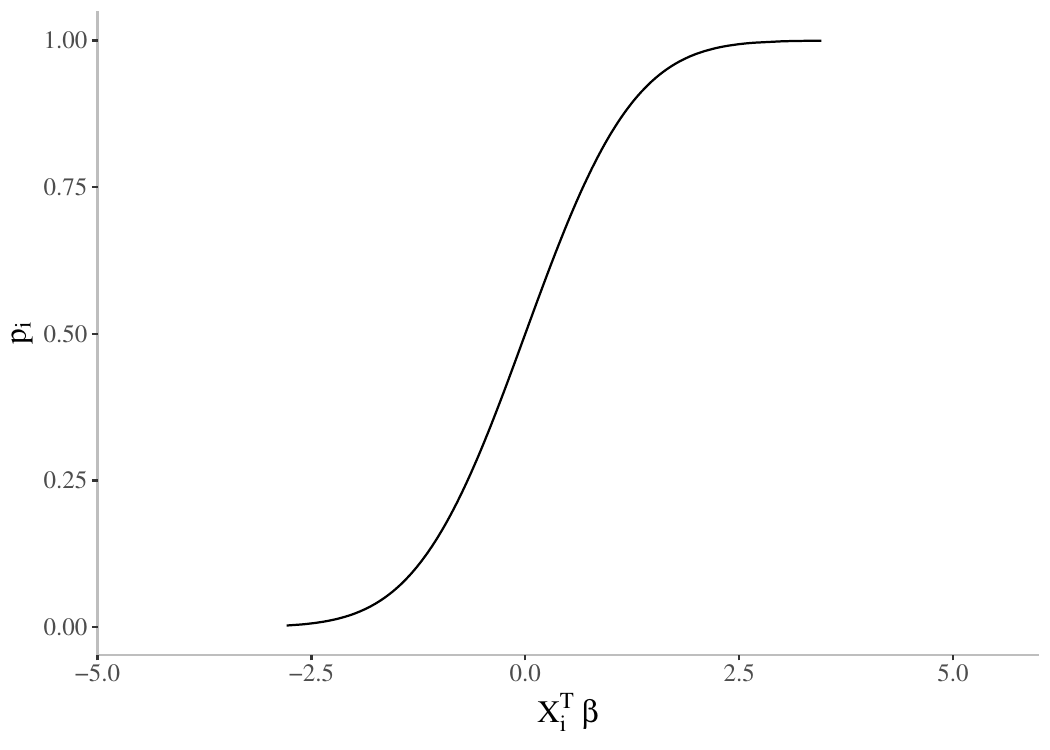}
         \caption{Probit}
     \end{subfigure}
     \begin{subfigure}[b]{0.49\textwidth}
         \centering
         \includegraphics[width=\textwidth]{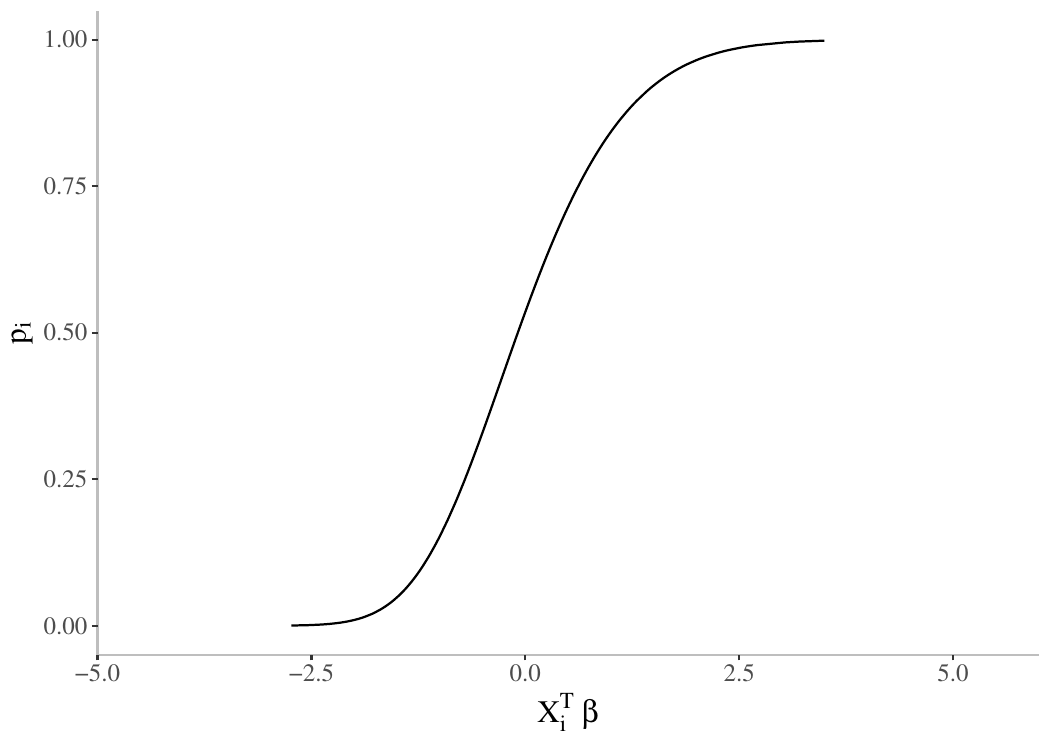}
         \caption{CSN}
     \end{subfigure}
     \begin{subfigure}[b]{0.49\textwidth}
         \centering
         \includegraphics[width=\textwidth]{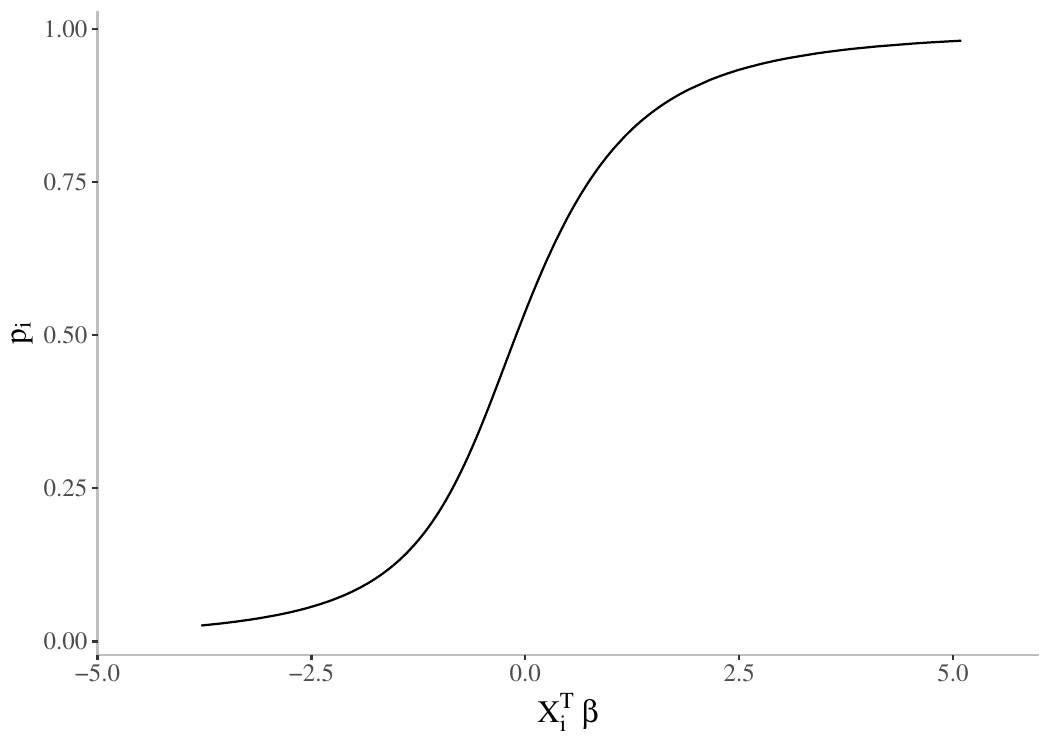}
         \caption{CST}
     \end{subfigure}
     \begin{subfigure}[b]{0.49\textwidth}
         \centering
         \includegraphics[width=\textwidth]{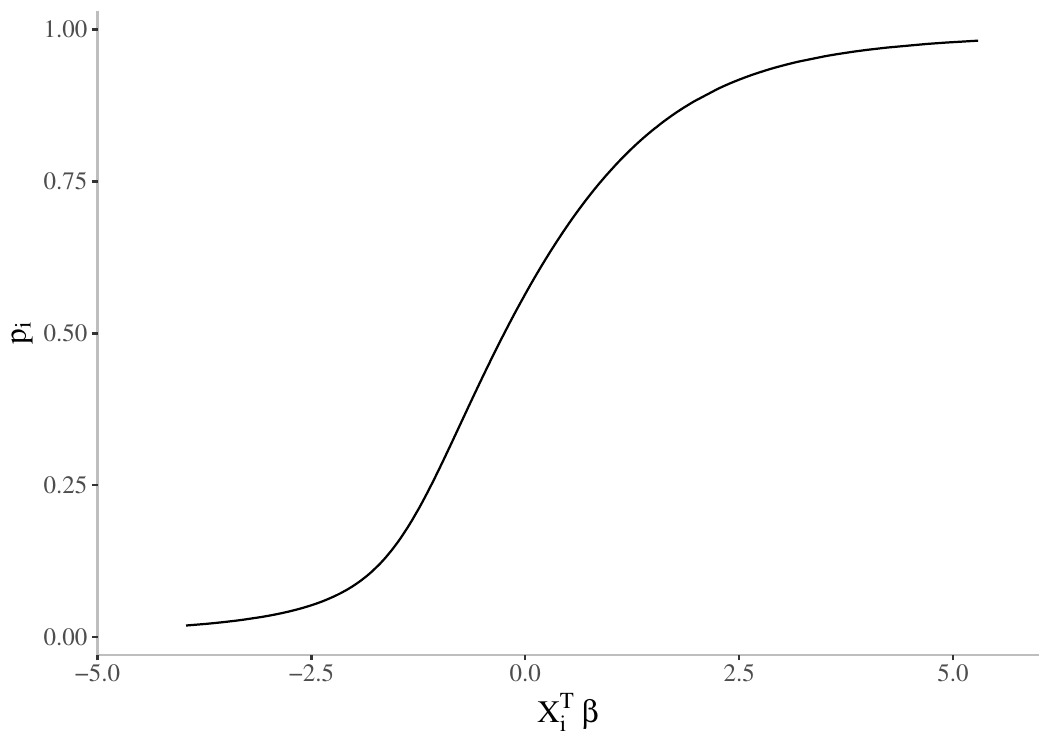}
         \caption{CSS}
     \end{subfigure}
     \begin{subfigure}[b]{0.49\textwidth}
         \centering
         \includegraphics[width=\textwidth]{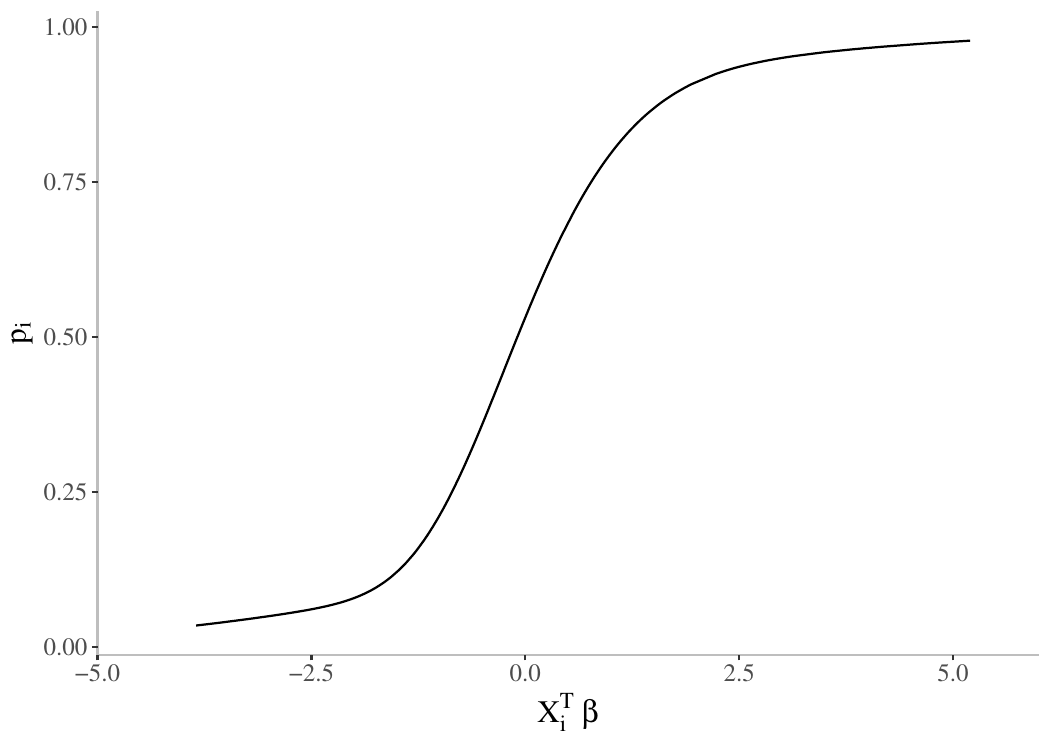}
         \caption{CSCN}
     \end{subfigure}
        \caption{Probability of heart disease as a function of $\eta_i$ for the fitted models.}
        \label{prob_success}
\end{figure}

\section{Conclusions}
In this paper we proposed a new class of link functions based on the SMCSN distributions. This class of link functions include symmetrical, skewed and robust link functions. We performed Bayesian estimation using latent variables to described the binary model. A strategy to select and fix the sign of the skewness was showed to avoid identifiability problems. Residual analysis was presented, and simulation studies were performed evaluating parameter recovery. The simulation study showed initially some problems in the accuracy of the estimates of $\bm{\beta}$, then we proposed to use a hyper-g prior to reduce the bias. Also, we noticed that for all sample size established, the estimates of all parameter tend to be closer to real values. An application was made on the heart disease data that indicated that the skewed and heavy-tailed link functions were preferred to the usual probit and CSN link models.

\section*{Acknowledgements}
    The authors wish to thank the Fundação de Amparo à Pesquisa do Estado de São Paulo (FAPESP, grant number 2015/25867-2) for the financial support and the partial financial support of CNPq.

\bibliographystyle{elsarticle-harv}

\bibliography{references}

\end{document}